\documentclass[11pt]{article}
\usepackage{fullpage}
\usepackage{times}
\usepackage{epsfig}

\usepackage[letterpaper,hmargin=1in,vmargin=1in]{geometry}

\newcommand{\qed}{\hfill\rule{2mm}{2mm}}
\newenvironment{proof}{\begin{trivlist}\item[\hspace{\labelsep}{\bf\noindent Proof. }]}{\end{trivlist}}

\newcommand{\co}{{\tt cost}}
\newcommand{\vc}{{\tt vc}}
\newcommand{\wgt}{{\tt wgt}}

\newtheorem{theorem}{Theorem}
\newtheorem{lemma}[theorem]{Lemma}
\newtheorem{claim}[theorem]{Claim}
\newtheorem{cor}[theorem]{Corollary}
\newtheorem{defn}[theorem]{Definition}

\renewcommand{\H}{{\cal H}}

\title{Enforcing efficient equilibria in network design games via subsidies\thanks{This work was partially supported by the grant NRF-RF2009-08 ``Algorithmic aspects of coalitional games'' and by the EC-funded STREP project EULER.}}
\author{John Augustine\thanks{Division of Mathematical Sciences, School of Physical and Mathematical Sciences, Nanyang Technological University, Singapore. Email: {\tt jea@ics.uci.edu}} \and Ioannis Caragiannis\thanks{Research Academic Computer Technology Institute \& Department of Computer Engineering and Informatics, University of Patras, 26500 Rio, Greece. Email: {\tt caragian@ceid.upatras.gr}}  \and Angelo Fanelli\thanks{Division of Mathematical Sciences, School of Physical and Mathematical Sciences, Nanyang Technological University, Singapore. Email: {\tt angelo.fanelli@ntu.edu.sg}} \and Christos Kalaitzis\thanks{Department of Computer Engineering and Informatics, University of Patras, 26500 Rio, Greece. Email: {\tt kalaitzis@ceid.upatras.gr}} }

\begin{document}

\maketitle

\begin{abstract}
The efficient design of networks has been an important engineering task that involves challenging combinatorial optimization problems. Typically, a network designer has to select among several alternatives which links to establish so that the resulting network satisfies a given set of connectivity requirements and the cost of establishing the network links is as low as possible. The {\sc Minimum Spanning Tree} problem, which is well-understood, is a nice example.

In this paper, we consider the natural scenario in which the connectivity requirements are posed by selfish users who have agreed to share the cost of the network to be established according to a well-defined rule. The design proposed by the network designer should now be consistent not only with the connectivity requirements but also with the selfishness of the users. Essentially, the users are players in a so-called network design game and the network designer has to propose a design that is an equilibrium for this game. As it is usually the case when selfishness comes into play, such equilibria may be suboptimal. In this paper, we consider the following question: can the network designer enforce particular designs as equilibria or guarantee that efficient designs are consistent with users' selfishness by appropriately subsidizing some of the network links? In an attempt to understand this question, we formulate corresponding optimization problems and present positive and negative results.
\end{abstract}

\section{Introduction}
Network design is a rich class of combinatorial optimization problems that model important engineering questions arising in modern networks. In an ideal scenario, a network designer that acts on behalf of a central authority is given an edge-weighted graph representing the potential links between nodes and their operation cost, and connectivity requirements between the nodes. The objective of the network designer is to compute a subgraph (the network to be established) of minimum cost that satisfies all connectivity requirements. Depending on the structure of the connectivity requirements, this definition leads to many optimization problems ranging from problems that are well-understood and efficiently solvable such as the {\sc Minimum Spanning Tree} to problems whose optimal solutions are even hard to approximate.

In this paper, we consider the scenario in which users are selfish and have agreed to a well-defined rule according to which they will share the cost of the network to be established. The connectivity requirements are now posed by the users; each user wishes to connect two specific nodes. A design should satisfy each connectivity requirement through a path connecting these two nodes in the established network. According to the particular cost sharing rule we consider, the corresponding user will then share the cost of each link in her path with the other users that use this link. Even though the network designer can still resort to the rich toolset of network design algorithms in order to propose a network of reasonable cost, this approach neglects the selfish behavior of the users. A user may not be satisfied with the current design since a different path that satisfies her connectivity requirement may cost her less. Then, she could unilaterally propose an alternative path that possibly includes links that were not in the proposal of the network designer. Other users could also act similarly and these negotiations compute the network to be established in a chaotic manner. The role of the network designer is almost canceled and, furthermore, it is not clear when the selfish users will reach an agreement (if they ever do) and, even if they do, whether this agreement will be really beneficial for the users as a whole, i.e., whether the total (or social) cost of the established network will be reasonable. So, the goal of the network designer is to propose a design (i.e., a network and, subsequently, a path to each user and an associated cost) that not only meets the connectivity requirements of the users but is also consistent with their selfish nature. Furthermore, since the network designer acts on behalf of the central authority, the design should be efficient, i.e., the network to be established should have reasonable social cost. Essentially, the users are engaged as players in a non-cooperative strategic game, called a {\em network design game}, and the role of the network designer is to propose an efficient design that is an equilibrium of this game.

Typically, efficiency is not an easy goal when selfishness comes into play. This leads to the following question which falls within one of the main lines of research in {\em Algorithmic Game Theory}: how is the social cost affected by selfish behavior? The notion of the {\em price of anarchy} (introduced in the seminal paper of Koutsoupias and Papadimitriou \cite{KP99}; see also \cite{P01}) can quantify this relation. Expressed in the context of a network design game, it would be defined as the ratio of the social cost of the worst possible Nash equilibrium over the social cost of an optimal design. Hence, it is pessimistic in nature and (as its name suggests) provides a worst-case guarantee for conditions of total anarchy. Instead, the notion of the {\em price of stability} that was introduced by Anshelevich et al. \cite{ADK+04} is optimistic in nature and quantifies how easy the job of the network designer is. It is defined as the ratio of the social cost of the best equilibrium over the cost of the optimal design and essentially asks: what is the best one can hope from a design given that the players are selfish?

Unfortunately, the price of stability can be large which would mean that every design that is consistent with selfishness has high social cost. The central authority could then intervene in order to mitigate the impact of selfishness. One solution that seems natural would be to contribute to the social cost of the network to be established by partially subsidizing some of the network links.
According to this scenario, the network designer has to compute a design and decide which links in the established network should be subsidized by the central authority. The users will then share the unsubsidized portion of the cost of the network links they use. Essentially, they will be involved in a new network design game and the goal of the network designer should be to guarantee that the design and the subsidies computed induce an equilibrium for this new game. Let us take this approach to its extreme in order to show that it is a feasible one. The network designer simply computes a design of low social cost (ignoring the issue of selfishness) which is fully subsidized by the central authority. The cost of each user is now zero and the design is obviously consistent with their selfishness. The problem becomes non-trivial when the central authority runs on a limited budget. What is the best design the network designer can guarantee given this budget? Alternatively, what is the minimum amount of subsidies sufficient in order to achieve a given social cost? Can optimality be achieved? Can the corresponding designs be computed efficiently?

\medskip\noindent{\bf Problem statement.} In an attempt to understand these questions, we introduce and study two optimization problems that arise in this context. Informally, they can be defined as follows. In {\sc Stable Network Enforcement} (SNE), we are given a network design game on a graph together with a particular target network $T$, and we wish to compute the minimum amount of subsidies that have to be put on the links of $T$ so that the design is acceptable to the users. In {\sc Stable Network Design} (SND), we are given a particular budget together with the input game, and we wish to compute a network $T$ that satisfies the connectivity requirements and to assign an amount of subsidies to the links of $T$ within the stipulated budget so that the design is acceptable to the users. The objective is to minimize the social cost of $T$. Besides the standard version of both problems, we also consider their {\em all-or-nothing} version in which a link can either be fully subsidized or not subsidized at all.

Even though some of our results apply to general network design games, we have placed emphasis on a special class of network design games, called {\em broadcast games}. In such a game, there is a special node in the input graph called the {\em root}. There is one player associated with each distinct non-root node and her connectivity requirement is a path from her associated node to the root. A nice property of such games is that an optimal design is a solution of the {\sc Minimum Spanning Tree} problem on the input graph and can be computed efficiently. Even in this seemingly simple case, as we will see, selfish behavior of players imposes challenging restrictions. Furthermore, we consider network design games in undirected graphs. Note that this strengthens our results since they can be adapted easily to network design games on directed graphs and, furthermore, undirected network design games are less understood in terms of their price of stability (see the discussion below).

\medskip\noindent{\bf Related work.}
Strategic games that arise from network design scenarios have received much attention in the {\em Algorithmic Game Theory} literature. The first related paper is probably \cite{ADTW03}. The particular network design games that we consider in the current paper were introduced by Anshelevich et al. in \cite{ADK+04}. An important observation made there is that network design games admit a potential function that was proposed by Rosenthal \cite{R73} for a broader class of games called congestion games. A potential function over all designs has the property that the difference in the potential of two designs that differ in the strategy of a single player equals the difference of the cost of that player in these designs; hence, a design that locally minimizes the potential function is a Nash equilibrium. Using a simple but elegant argument, Anshelevich et al. \cite{ADK+04} proved that the price of stability is at most $\H_n$, the $n$-th harmonic number, where $n$ is the number of players. Their proof considers a Nash equilibrium that can be reached from an optimal design when the players make arbitrary selfish deviations. The main argument used is that the potential of the Nash equilibrium is strictly smaller than that of the optimal design and the proof follows due to the fact that the potential function of Rosenthal approximates the social cost of any design within a factor of at most $\H_n$. Much of the subsequent research on network design games has focused on providing tight bounds on the price of stability. The $\H_n$ bound is known to be tight for directed networks only. For undirected networks, better bounds on the price of stability of $O(\log\log n)$ for broadcast games and $O(\log n /\log\log{n})$ for generalizations known as multicast games are presented in \cite{Li09} and \cite{FKL+06}, respectively. Still, the best lower bounds are only constant; $2.245$ in general, $1.862$ for multicast, and $1.818$ for broadcast games \cite{BCFM10}. The papers \cite{CCL+07} and \cite{CKM+08} provide bounds on the quality of equilibria reached when players enter a multicast game one by one and play their best response and then (when all players have arrived) they concurrently play until an equilibrium is reached. They prove that the price of anarchy of the equilibria reached is at most polylogarithmic in the number of players.

Another intriguing question is related to the complexity of computing equilibria in such games. In general, the problem was recently proved to be PLS-hard \cite{S10}. The corresponding hardness reduction does not apply to multicast or broadcast games. Unfortunately, the classical approach of minimizing the potential function that has been proved useful in the case of network congestion games \cite{FPT04} cannot be applied to multicast games; the authors of \cite{CCL+07} prove that minimizing Rosenthal's potential function is NP-hard. Furthermore, as it is observed in \cite{S10}, computing an equilibrium of minimum social cost in multicast games is NP-hard. Approximate equilibria is the subject of the recent paper \cite{AL10}.

Monetary incentives in strategic games have been considered in many different contexts. Most of the work in {\em Mechanism Design} uses such incentives to motivate players to act truthfully (see \cite{N07} for an introduction to the field). The (non-exhaustive) list also includes their use in {\em Cooperative Game Theory} in order to encourage coalitions of players reach stability \cite{BEM+09} and as a means to stabilize normal form games \cite{MT04}. However, the particular use of monetary incentives in the current paper is substantially different and aims at improving the performance of the system the game represents. In this direction, other tools have also been considered recently. Besides their importance in creating income, the appropriate use of {\em taxes} can also improve system efficiency. A series of recent papers \cite{CDR06,CKK10,FJM04,KK04,S07} study the impact of taxes (or tolls) in the efficiency of network routing (extending early developments in the literature of the Economics of Transportation; see \cite{CDR06} and the references therein). {\em Stackelberg strategies} applied to routing \cite{KLO97}, scheduling \cite{KM02,R04,S07}, and, recently, network design games \cite{FFM10} aim to improve performance by introducing a set of non-selfish leaders whose strategies are controlled by the system designer and aim to motivate selfish players to reach efficient equilibria. {\em Coordination mechanisms} (applied so far to scheduling in \cite{AJM08,C09,CKN04,ILMS05}) aim to improve performance by introducing priorities among the players. An approach that is closer to ours has been followed in \cite{BLNO10} where subsidies are used in multicast games; unlike our approach, the subsidies are collected as taxes from the players in order to guarantee efficient worst-case equilibria.

\medskip\noindent{\bf Overview of results and roadmap.} In this paper we present the following results. First, we observe that SNE can be solved in polynomial time using linear programming; this observation applies to general instances of SNE. The linear program has an exponential number of constraints but can be solved using the ellipsoid method. A reformulation based on standard techniques yields an LP of polynomial size. For instances of SNE with broadcast games, we present a much simpler LP in which the number of variables and constraints is linear and quadratic in the number of players, respectively. On the other hand, SND is proved to be NP-hard even for broadcast instances. In particular, detecting whether a minimum spanning tree can be enforced as an equilibrium without using any subsidies is NP-hard. This result implies that detecting whether the price of stability of a given broadcast game is $1$ or not is NP-hard. In this direction, we have a stronger result: approximating the price of stability of a broadcast game is APX-hard. The last two statements significantly extend the NP-hardness result of \cite{S10} and indicate that, besides the rough estimates provided by the known bounds on the price of stability which hold for a broad class of games, the estimate the network designer can make about the most efficient designs of a particular broadcast game will also be rough. These results are presented in Section \ref{sec:complexity}.

Next, we consider broadcast instances of SNE and the question of how much subsidies are sufficient and necessary in order to enforce a given minimum spanning tree as an equilibrium. We show that this can be done using a percentage of $37\%$ of the weight of the minimum spanning tree as subsidies. The proof has two main components. First, we show how to prove this upper bound by decomposing the game into subgames with a significantly simpler structure than the original one. Second, in order to compute the subsidies in each subgame, we use a virtual approximation of the cost experienced by the players on the links of the network. We also demonstrate that our upper bound is tight: an amount of $37\%$ of the minimum spanning tree weight as subsidies may be necessary for some simple instances. These results are presented in Section \ref{sec:bounds}.

Surprisingly, in contrast to the standard case, we prove that the all-or-nothing version of SNE is hard to approximate within any factor even when restricted to instances with broadcast games. The corresponding proof is long and technically involved and indicates that the only approximation guarantee should bound the amount of subsidies as a constant fraction of the weight of the minimum spanning tree. Interestingly, we prove that significantly more subsidies may be necessary compared to the standard version of SNE. In particular, there are broadcast instances which require a percentage of $61\%$ of the weight of the minimum spanning tree as subsidies in order to enforce it as an equilibrium. These results are presented in Section \ref{sec:all-or-nothing}.

We begin with preliminary definitions and notation in Section \ref{sec:prelim} and conclude with interesting open problems in Section \ref{sec:open}.

\section{Definitions and notation}\label{sec:prelim}
A {\em network design game} consists of an edge-weighted undirected graph $G=(V,E,w)$, a set $N$ of $n$ players, and a source-destination pair of nodes $(s_i,t_i)$ for each player $i$. Each player wishes to connect her source to her destination and, in order to do this, she can select as a strategy any path $T_i$ connecting $s_i$ to $t_i$ in $G$. The tuple $T=(T_1, T_2, ..., T_n)$ that consists of the strategies of the players (with one strategy per player) is called a {\em state}. We say that player $i$ uses edge $a$ in $T$ if her strategy $T_i$ contains $a$. With some abuse in notation, we also denote by $T$ the set of edges included in strategies $T_1, T_2, ..., T_n$ as well as the subgraph of $G$ induced by these edges. We say that an edge $a\in E$ is established if at least one player uses edge $a$. Consider such an edge $a$ and let $n_a(T)$ be the number of players whose strategies in $T$ contain $a$. Throughout the paper, we also use the notation $n^i_a(T)$ to denote whether player $i$ uses edge $a$ ($n_a^i(T)=1$) or not ($n_a^i(T)=0$). Each player $i$ in $N$ experiences a cost of $\co_i(T)=\sum_{a\in T_i}{\frac{w_a}{n_a(T)}}$, i.e., the weight of each established edge is shared as cost among the players using it.

The state $T$ is called a (pure Nash) equilibrium if no player has an incentive to unilaterally deviate from $T$ in order to decrease her cost, i.e., for each player $i$ and possible strategy $T'_i$ that connects the source-destination pair $(s_i,t_i)$ in $G$, it holds that $\co_i(T)\leq \co_i(T_{-i},T'_i)$. The notation $T_{-i},T'_i$ denotes the state in which player $i$ uses strategy $T'_i$ and the remaining players use their strategies in $T$. Throughout the paper, we denote by $\wgt(A)$ the total weight of the set of edges $A$ in $G$, i.e., $\wgt(A)=\sum_{a\in A}{w_a}$. The quality of a state is measured by the total weight of the established edges. Since the weight of each established edge is shared as cost among the players that use it, the quality of a state coincides with the total cost experienced by all players, i.e., $\wgt(T)=\sum_i{\co_i(T)}$. The {\em price of stability} of a network design game is simply the ratio of the weight of the edges established in the best equilibrium over the optimal cost among all states of the game.

Given an edge-weighted graph $G=(V,E,w)$, a subsidy assignment $b$ is a function that assigns a subsidy $b_a\in [0,w_a]$ to each edge $a\in E$. The cost of a subsidy assignment is simply the sum of the subsidies on all edges of $G$, i.e., $\sum_{a\in E}{b_a}$. We use the term {\em all-or-nothing} to refer to subsidies that are constrained so that $b_a\in \{0,w_a\}$ for each edge $a\in E$. Given a set of edges $A$ in $G$, we use the notation $b(A)$ in order to refer to the total amount of subsidies assigned to the edges of $A$ in the subsidy assignment $b$, i.e., $b(A)=\sum_{a\in A}{b_a}$. We refer to $b(E)$ as the cost of the subsidy assignment $b$. Given a network design game on a graph $G$ and a subsidy assignment $b$ on the edges of $G$, we use the term {\em extension} of the original game with subsidies $b$ in order to refer to the network design game on graph $G$ (with the same players and strategy sets as in the original game) with the only difference being that the cost of a player at a state $T$ is now $\co_i(T;b) = \sum_{a\in T_i}{\frac{w_a-b_a}{n_a(T)}}$. When a particular state $T$ is an equilibrium of the extension of the original game with subsidies $b$, we say that the subsidy assignment $b$ {\em enforces} $T$ as an equilibrium in the extension of the original game.

An instance of the {\sc Stable Network Design} problem (SND) consists of a network design game on a graph $G$, a budget $B$, and a positive number $K$. The question is whether there exists a subsidy assignment $b$ of cost at most $B$ on the edges of $G$ so that a subgraph of $G$ of total weight at most $K$ is an equilibrium for the extension of the original game with subsidies $b$. An instance of the {\sc Stable Network Enforcement} problem (SNE) consists of a network design game on a graph $G$, a budget $B$, and a state $T$. The question is whether there exists a subsidy assignment of cost at most $B$ on the edges of $G$ so that $b$ enforces $T$ as an equilibrium on the extension of the original game with subsidies $b$. Note that the subsidy assignment does not need to put any subsidies to edges not in $T$. In the integral versions of SNE and SND, the subsidy assignment in question is all-or-nothing. Of course, optimization versions of the above problems are natural. For example, in an optimization version of SNE, we are given the network design game on a graph $G$ and a state $T$, and we require the subsidy assignment in question to be of minimum cost.

Broadcast games are special cases of network design games. In a broadcast game, the graph $G$ has exactly $n+1$ nodes; all players have the same destination node, which is called the {\em root} and is denoted by $r$, and distinct non-root nodes as sources. In such games, we refer to a player with a source node $u$ as the player associated with node $u$ (and use $u$ to identify the player). Clearly, any state $T$ in such a game spans all nodes of $G$ and a minimum spanning tree is a state that minimizes the total cost experienced by the players. Given any spanning tree $T$ and a non-root node $u$, we denote by $T_u$ the path from $u$ to $r$ in $T$. In broadcast games, we mostly consider equilibria that are spanning trees. It can be easily seen that if an equilibrium $T$ contains a cycle, then all edges in this cycle must have zero weight; then, there is an alternative spanning tree with the same total weight that is also an equilibrium.

\section{The complexity of SNE and SND}\label{sec:complexity}
We begin the presentation of our results with the following observation.

\begin{theorem}\label{thm:SNE-P}
{\sc Stable Network Enforcement} is in P.
\end{theorem}
This theorem applies to general instances of SNE. By slightly deviating from the main focus of the paper which is on broadcast games, we first discuss how general instances of SNE can be expressed using linear programming; analogous formulations have been proposed for the computation of tolls in non-atomic selfish routing games (e.g., see \cite{FJM04}). Then, specifically for broadcast SNE instances, we present a much simpler LP formulation.

We will describe a linear program which, given a network design game with a set $N$ of $n$ players on an edge-weighted graph $G=(V,E,w)$ and a state $T$, solves the optimization version of the problem by computing a subsidy assignment $b$ of minimum cost so that $T$ is an equilibrium of the extension of the original game on $G$ with subsidies $b$. The subsidies in question are the variables of the LP. We remark that, even though we only need to use variables for the subsidies on the edges of $T$, we assume that $b_a$ is defined for each edge $a$ of $E$ in order to simplify the presentation; it should be clear that, in any optimal solution of the linear programs below, $b_a=0$ for each edge $a\in E\setminus T$. The variables are constrained so that $b_a\in [0,w_a]$ while there are constraints that capture the requirement that $T$ is an equilibrium in the extension of the original game with subsidies $b$. For each player $i$, this means that the cost the player experiences in $T$ should not be higher than the cost she would experience when deviating to any other strategy $T'_i$ (i.e., a path that connects her source node $s_i$ to her destination $t_i$). This is captured by the inequality in the second line of LP (\ref{lp:1}), where ${\cal T}_i$ denotes the set of all paths connecting node $s_i$ to node $t_i$ in $G$. The left hand side of the inequality is simply the cost $\co_i(T;b)$ of player $i$ in $T$ while the right hand side is the cost $\co_i(T_{-i},T'_i;b)$ she would experience in state $T_{-i},T'_i$. Observe that the denominator $n_a(T)+1-n^i_a(T)$ equals the number of players using edge $a$ in $T_{-i},T'_i$. Also note that, given $T$ and $i$, the quantities $n_a(T)$ and $n_a^i(T)$ are fixed and, clearly, all constraints are linear.
\begin{eqnarray}\label{lp:1}
\mbox{minimize} & &\sum_{a\in E}{b_a}  \\\nonumber
\mbox{subject to} & & \forall i\in N, T'_i\in {\cal T}_i, \sum_{a\in T_i}{\frac{w_a-b_a}{n_a(T)}}\leq \sum_{a\in T'_i}{\frac{w_a-b_a}{n_a(T)+1-n^i_a(T)}}\\\nonumber
& & \forall a\in E, 0\leq b_a\leq w_a
\end{eqnarray}
In general, the above LP has an exponential number of constraints (one for each player $i$ and each path in ${\cal T}_i$) but can be solved in polynomial time using the ellipsoid method (see \cite{GLS93}). All that is needed is a {\em separation oracle} which returns a violating constraint (if one exists) for a given subsidy assignment $b$ in polynomial time. We demonstrate how this can be done for the constraints associated with player $i$. We construct an edge-weighted graph $H_i$ over the set of nodes $V$ (and set of edges $E$) so that the weight of edge $a$ is defined as $w'_a=\frac{w_a-b_a}{n_a(T)+1-n^i_a(T)}$. We compute a shortest path $p_i$ from node $s_i$ to node $t_i$ in $H_i$. If the length of path $p_i$ satisfies $\sum_{a\in p_i}{w'_a}<\co_i(T;b)$, then the constraint that is associated with path $p_i$ for player $i$ is violated. Otherwise, no constraint associated with player $i$ is violated.

We can transform the above LP to an equivalent one that has polynomial size. The main idea is to simulate the separation oracle for the above LP using additional variables and constraints. For each player $i$ and node $v$ of $G$, we introduce the variable $\pi_i(v)$ to denote a lower bound on the length of the shortest path in graph $H_i$ from node $s_i$ to node $v$. The first two lines in the constraints of the following LP guarantee that $\pi_i(v)$ is indeed such a lower bound (in the first constraint, $\Gamma(u)$ denotes the set of neighbors of node $u$ in $G$). Then, the constraint $\pi_i(t_i)\geq \co_i(T;b)$ guarantees that the player $i$ has no incentive to deviate from her strategy in $T$ in the extension of the original game with subsidies $b$.
\begin{eqnarray}\label{lp:2}
\mbox{minimize} & &\sum_{a\in E}{b_a}  \\\nonumber
\mbox{subject to} & & \forall i\in N, u\in V, v\in \Gamma(u), \pi_i(v)\leq \pi_i(u)+\frac{w_{(u,v)}-b_{(u,v)}}{n_{(u,v)}(T)+1-n^i_{(u,v)}(T)}\\\nonumber
& & \forall i\in N, \pi_i(s_i) = 0\\\nonumber
& & \forall i\in N, \pi_i(t_i)\geq \sum_{a\in T_i}{\frac{w_a-b_a}{n_a(T)}}\\\nonumber
& & \forall i\in N, u\in V\setminus \{s_i\}, \pi_i(u) \geq 0\\\nonumber
& & \forall a\in E, 0\leq b_a\leq w_a
\end{eqnarray}

LP (\ref{lp:2}) has $\Theta(n|V|)$ variables and $\Theta(n|E|)$ constraints. We have a much simpler LP when the input is an instance of SNE consisting of a broadcast game on graph $G$ (with a root node $r$) and a spanning tree $T$ of $G$. We use the same variables as in the original LP (i.e., $n$ variables since $T$ is a spanning tree over $n+1$ nodes now) and much fewer (i.e., $O(|E|)$) constraints. In particular, we just require that no player associated with a node $u$ has an incentive to change her strategy in $T$ and use an edge $(u,v)$ that does not belong to $T$ and the path from $v$ to $r$ in $T$. The corresponding LP is:
\begin{eqnarray}\label{lp-simple}
\mbox{minimize} & &\sum_{a\in E}{b_a}  \\\nonumber
\mbox{subject to} & & \forall u\in V \setminus \{r\}, v\in \Gamma(u) \mbox{ such that $(u,v)\not\in T$},\\\nonumber
 & & \quad\sum_{a\in T_u}{\frac{w_a-b_a}{n_a(T)}}\leq w_{(u,v)}+\sum_{a\in T_v}{\frac{w_a-b_a}{n_a(T)+1-n^u_a(T)}}\\\nonumber
& & \forall a\in E, 0\leq b_a\leq w_a
\end{eqnarray}
The correctness of the above LP (i.e., its equivalence with the optimization version of SNE) is given by the following lemma.
\begin{lemma}\label{lem:lp-simple}
Consider an instance of {\sc Stable Network Enforcement} consisting of a broadcast game on a graph $G$ and a state $T$. A subsidy assignment $b$ enforces $T$ as an equilibrium in the extension of the broadcast game in $G$ if and only if the constraints of LP (\ref{lp-simple}) are satisfied.
\end{lemma}
\begin{proof}
If $T$ is an equilibrium of the extension of the broadcast game on $G$ with subsidies $b$, then clearly the constraints of LP (\ref{lp-simple}) are satisfied since otherwise a player associated with node $u$ would have an incentive to deviate and use the path that consists of the edge from $u$ to another node $v$ that is not part of $T$ and the path from $v$ to $r$ in $T$ (or, simply, of the direct edge from $u$ to the root node if $v=r$).

Now, assume that all constraints of LP (\ref{lp-simple}) are satisfied for a subsidy assignment $b$. We will show that no player has an incentive to deviate from $T$ in the extension of the broadcast game with subsidies $b$. Assume for the sake of contradiction that the player associated with node $u$ has an incentive to deviate to some path from $u$ to $r$ that includes edges not belonging to $T$. Among such paths, let $p$ be the path that incurs the minimum cost for the player associated with node $u$ and, furthermore, contains the minimum number of edges not belonging to $T$. Following the edges of path $p$ from node $u$ to node $r$, let $(v_1,v_2)$ be its last edge that does not belong to $T$. Using the constraint of the LP for the player associated with node $v_1$ and edge $(v_1,v_2)$, we have that
\begin{eqnarray*}
\sum_{a\in T_{v_1}}{\frac{w_a-b_a}{n_a(T)}} &\leq & w_{(v_1,v_2)}+\sum_{a\in T_{v_2}}{\frac{w_a-b_a}{n_a(T)+1-n_a^{v_1}(T)}}.
\end{eqnarray*}
Also, let $v_3$ be the least common ancestor of $v_1$ and $v_2$ in $T$ and denote by $q_1$ and $q_2$ the subpaths of $T_{v_1}$ and $T_{v_2}$ that connect $v_1$ and $v_2$ to $v_3$, respectively. Since $T_{v_1}$ and $T_{v_2}$ use the same edges in order to connect $v_3$ to $r$ and $q_1$ and $q_2$ are edge-disjoint (and, hence, $n_a^{v_1}(T)=0$ for each $a\in q_2$), we also have that
\begin{eqnarray}\label{eq:contradiction-1}
\sum_{a\in q_1}{\frac{w_a-b_a}{n_a(T)}} &\leq &w_{(v_1,v_2)}+\sum_{a\in q_2}{\frac{w_a-b_a}{n_a(T)+1}}.
\end{eqnarray}
Now, observe that the cost experienced by player $u$ when deviating to the path $p$ is strictly smaller than the cost she would experience by using the path $p'$ consisting of the subpath of $p$ connecting $u$ to $v_1$, the edges of $q_1$, and the edges of $p$ from $v_3$ to $r$. Otherwise, the path $p'$ would either incur strictly smaller cost to the player associated with node $u$ than $p$ or it would also incur the same cost as $p$ but it would have strictly fewer edges not belonging to $T$; both cases contradict our assumptions about $p$. Hence,
\begin{eqnarray*}
w_{(v_1,v_2)}+\sum_{a\in q_2}{\frac{w_a-b_a}{n_a(T)+1-n^u_a(T)}} &<& \sum_{a\in q_1}{\frac{w_a-b_a}{n_a(T)+1-n^u_a(T)}},
\end{eqnarray*}
which implies that
\begin{eqnarray}\label{eq:contradiction-2}
w_{(v_1,v_2)}+\sum_{a\in q_2}{\frac{w_a-b_a}{n_a(T)+1}} &<& \sum_{a\in q_1}{\frac{w_a-b_a}{n_a(T)}},
\end{eqnarray}
since $n_a^u(T)\in\{0,1\}$. We have reached a contradiction between (\ref{eq:contradiction-1}) and (\ref{eq:contradiction-2}) and the lemma follows.\qed\end{proof}

Next, we prove that the restriction of SND to broadcast instances is NP-hard. The hardness proof uses instances of SND with budget equal to zero with target equilibrium weight equal to the weight of the minimum spanning tree. Note that in instances with a unique minimum spanning tree, the problem is certainly in P; one can just compute a minimum spanning tree and apply the LP approach described above. In our reduction, there are many different minimum spanning trees but it is hard to detect whether there is one that is an equilibrium in the corresponding broadcast game.

\begin{theorem}\label{thm:broadcastNPC}
Given an instance of {\sc Stable Network Design} consisting of a broadcast game on a graph $G$, budget $B$, and a positive number $K$, it is NP-hard to decide whether there exists a subsidy assignment $b$ of cost at most $B$ so that the extension of the game with subsidies $b$ has a tree of weight at most $K$ as an equilibrium. Moreover, it is NP-hard even when $B$ is set to zero.
\end{theorem}

We will first describe a gadget that is used in the proof of Theorem \ref{thm:broadcastNPC}; we call it the {\tt Bypass} gadget of capacity $\kappa$. The gadget is shown in Figure~\ref{fig:single}. Let $\ell$ be the minimum positive integer such that $\H_{\kappa+\ell} - \H_{\kappa} > 1$. The {\tt Bypass} gadget consists of a root node $r$ connected to one end of a path of $\ell$ nodes formed with edges of unit weight. We call this the {\em basic path} of the {\tt Bypass} gadget. The node $c$ on the far end of the path from $r$ is called the {\em connector node}. There is an edge from $c$ to $r$ of weight $\H_{\kappa+\ell} - \H_{\kappa}$, which we call the {\em bypass edge}.

Suppose this gadget is connected to a subgraph $S$ of $\beta$ nodes as shown in Figure~\ref{fig:single}. For the moment, we are not concerned with how the nodes in $S$ are connected to each other. Consider the instance of SNE consisting of the broadcast game on the graph $G$ of Figure~\ref{fig:single}, budget $B=0$, and let $T$ be a minimum spanning tree of $G$. Note that $T$ does not include the bypass edge; it includes all edges in the basic path from $c$ to $r$ instead.

\begin{figure*}[htbp]
\centerline{\psfig{file=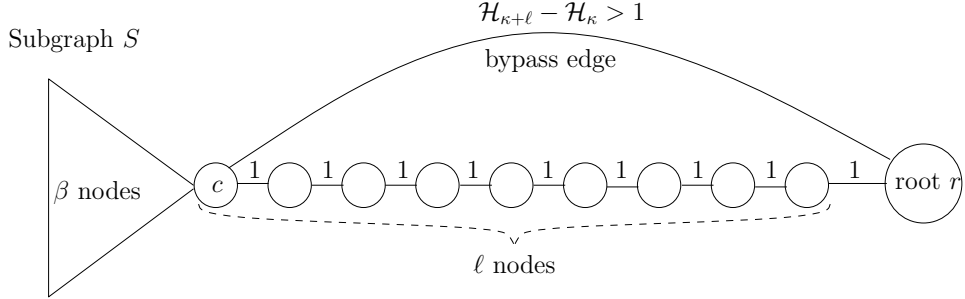,width=5in}}  
\caption{The {\tt Bypass} gadget with capacity $\kappa$.} \label{fig:single}
\end{figure*}

\begin{lemma}\label{lem:gadget}
If $\beta < \kappa$, then the player associated with node $c$ has an incentive to deviate from her strategy in $T$ and use the bypass edge. Otherwise, no player associated with a node in the basic path has any incentive to deviate from $T$.
\end{lemma}
\begin{proof}
Regardless of how the players associated with nodes in the subgraph $S$ are routed, since $S$ is connected to the {\tt Bypass} gadget through node $c$, there are $\beta+1$ players that need to use a path from $c$ to $r$. Let us focus on the player associated with node $c$. If she and all the $\beta$ players from $S$ take the basic path, then her cost will be $\sum_{i=1}^{\ell}\frac{1}{\beta + i} = \H_{\beta + \ell} - \H_{\beta}$. If $\beta < \kappa$, then $\H_{\kappa + \ell} - \H_{\kappa} < \H_{\beta + \ell} - \H_{\beta}$, and therefore, the player associated with node $c$ has an incentive to deviate to the bypass edge. On the other hand, if $\beta \ge \kappa$, then $\H_{\kappa + \ell} - \H_{\kappa} \ge \H_{\beta + \ell} - \H_{\beta}$ and any player associated with a node in the basic path experiences a cost of at most $\H_{\beta + \ell} - \H_{\beta}$ in $T$. Hence, no such player has an incentive to deviate from $T$.
\qed\end{proof}

\noindent {\bf Proof of Theorem \ref{thm:broadcastNPC}. ~}
We show that the problem is NP-hard even when we consider the special case where $B=0$ and $K$ equals the weight of the minimum spanning tree of the input graph $G$. In other words, given a broadcast game on a graph $G$ with root node $r$, we ask: does this game have a minimum spanning tree of $G$ as an equilibrium? We use a reduction from {\sc Bin Packing}.

We use a stricter form of {\sc Bin Packing} defined as follows. We are given a set of $n$ items indexed by $i \in \{1, 2, \ldots, n\}$. The size of each item $i$ is a positive even integer denoted by $s_i$. Since bin packing is strongly NP-hard \cite{GJ79}, we assume that $s_i$ is bounded by a polynomial in $n$. We are also given a set of $k$ bins indexed by $j\in \{1, 2, \ldots, k\}$, each of even integer capacity $C$, which we assume to be at least as large as $\max_{i}{s_i}$. We furthermore assume that $\sum_{i=1}^n s_i = k C$. We ask whether each item can be allocated to one of the $k$ bins so that the total size of items in each bin is exactly $C$. Our definition of {\sc Bin Packing} is somewhat stricter than the conventional definition in which the capacity of bins and the size of the items is not restricted to be even and bins are not required to be filled to the brim. However, we note that it is quite straightforward to see that this restricted version of the problem can be reduced from the conventional version by first adding a suitable number of unit-sized items and then doubling the size of all items and the capacity of all bins. The number of additional items is upper-bounded by the total capacity of the bins. Therefore, our restriction of {\sc Bin Packing} is also strongly NP-hard.

\begin{figure*}[htbp]
\centerline{\psfig{file=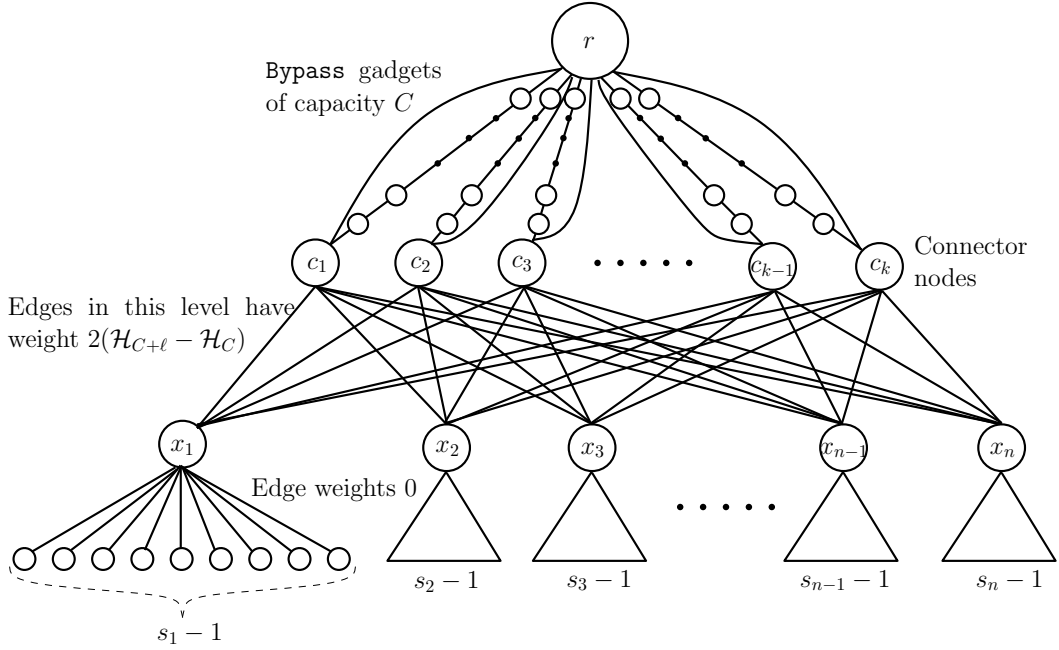,width=5.5in}}  
\caption{The graph $G$ constructed from an instance of {\sc Bin Packing}.} \label{fig:reduction}
\end{figure*}

Given a restricted instance of {\sc Bin Packing}, we now construct an instance of SNE as follows. For each item $i$ of size $s_i$, we create a star graph with one center node which we denote by $x_i$ and $s_i-1$ leaves. The edges connecting the leaves to the center node of the star have zero weight. Let $X$ be the set of center nodes. For each bin $j$, construct a {\tt Bypass} gadget with capacity $\kappa = C$. Again, let $\ell$ be the number of unit-weight edges in the basic path of each {\tt Bypass} gadget. Recall that $\ell$ is the minimum positive integer such that $\H_{C+\ell}-\H_C>1$; this implies that $\ell$ is linear in $C$. We denote the connector node in the gadget corresponding to bin $j$ by  $c_j$.  Let $\chi$ be the set of all connector nodes. We connect sets $\chi$ and $X$ by a complete bipartite edge set with edges having weight $2(\H_{C+\ell}-\H_C)$. Observe that any minimum spanning tree of $G$ consists of the $k\ell$ unit-weight edges in the {\tt Bypass} gadgets, the zero-weight edges connecting the leaves to their star center, and $n$ edges that connect nodes of $\chi$ to nodes of $X$ so that each node of $X$ is connected to exactly one node of $\chi$. We set $K$ to be the weight of the minimum spanning tree, i.e., $K = k \ell + 2n (\H_{C+\ell}-\H_C)$.

We claim that a minimum spanning tree $T_{ne}$ of $G$ is an equilibrium for the broadcast game on $G$ if and only if the {\sc Bin Packing} instance has a solution. We prove this claim in both directions.

Let $T_{ne}$ be a minimum spanning tree that is an equilibrium. Let $\beta_j +1$ be the number of nodes in the subtree of $T_{ne}$ rooted at $c_j$. From Lemma~\ref{lem:gadget}, we know that since $T_{ne}$ is an equilibrium, for all $j$, it holds that $\beta_j \ge C$. However, we also know from the properties of the {\sc Bin Packing} instance and the construction of graph $G$ that $\sum_{j=1}^k\beta_j=\sum_{i=1}^n s_i = k C$. Clearly, it follows that for all $j$, we have $\beta_j = C$. Therefore, the allocation of item $i$ to bin $j$ whenever $x_i$ is connected to $c_j$ will lead to a solution for the {\sc Bin Packing} instance since the total size of these items is exactly $\beta_j=C$.

To show the other direction, let us suppose that we have a solution to the {\sc Bin Packing} instance. We construct a minimum spanning tree $T_{ne}$ as follows. $T_{ne}$ contains the edges from the leaves to the corresponding star center, the basic paths from the connector nodes to the root node $r$, and the edge $(x_i, c_j)$ for each item $i$ that is allocated to bin $j$. Note that, for $j=1, ..., k$, the number of nodes in the subtree of $T_{ne}$ rooted at $c_j$ is exactly $C$. So, any player associated with a node in a basic path experiences a cost of at most $\H_{C+\ell}-\H_C$ in $T_{ne}$. Furthermore, observe that each edge of $T_{ne}$ between nodes of $\chi$ and $X$ is used by at least two players in $T_{ne}$. So, any player associated with a node in a star experiences a cost of at least $2(\H_{C+\ell}-\H_C)$. Hence, no player has an incentive to deviate to a path that includes a node of $\chi$ that she does not use in $T_{ne}$. Any such path would include an edge of weight $2(\H_{C+\ell}-\H_C)$ between a node in $\chi$ and a node in $X$ that is used only by that player. So, for each $j$, the $C$ players associated with nodes in the subtree of $c_j$ in $T_{ne}$ have no incentive to deviate to a path that does not use node $c_j$. By Lemma \ref{lem:gadget}, player $c_j$ (and, consequently, all players that have node $c_j$ in their path to $r$ in $T_{ne}$) has no incentive to deviate to the bypass edge connecting $c_j$ to $r$. This holds for any other node in the basic path of a {\tt Bypass} gadget as well. Therefore, it follows that $T_{ne}$ is an equilibrium.  \qed \medskip

Note that the proof of Theorem \ref{thm:broadcastNPC} essentially implies that deciding whether the price of stability of a given broadcast game is $1$ or not is NP-hard. The next statement provides an even stronger negative result. It implies that given instances of SND consisting of a broadcast game on a graph $G$ and a budget $B$, it is NP-hard to approximate within a factor better than $571/570$ the minimum weight among all equilibria in any extension of the original game with subsidies of cost at most $B$.

\begin{theorem}\label{thm:pos}
Approximating the price of stability of a broadcast game within a factor better than $571/570$ is NP-hard.
\end{theorem}

\begin{proof}
The proof is based on a reduction from {\sc Independent Set} in $3$-regular graphs and uses an inapproximability result due to Berman and Karpinski \cite{BK99}. Given a $3$-regular graph $H$ with $n$ nodes and $3n/2$ edges, we construct an instance of a broadcast game consisting of a graph $G$ as follows. The graph $G$ has a node for each node and each edge of $H$ and an additional root-node $r$. We denote by $U$ the set of nodes of $G$ that correspond to a node of $H$ and by $V$ the set of nodes of $G$ that correspond to an edge of $H$. For each non-root node of $G$, there is an edge connecting it with the root; these edges have unit weight. A node of $V$ that corresponds to an edge $(u,v)$ in $H$ is connected with edges to the nodes of $U$ that correspond to the nodes $u$ and $v$ of $H$. The weight of these edges is $\frac{2+\delta}{3}$ for some $\delta \in (0,1/12]$. Clearly, the subgraph of $G$ induced by the nodes in $U\cup V$ (i.e., all nodes besides $r$) is bipartite; in this subgraph, the nodes of $U$ have degree $3$ while the nodes of $V$ have degree $2$.

We claim that the graph $H$ has an independent set of size $m$ if and only if the broadcast game has an equilibrium of weight $5n/2-(1-\delta)m$. Consider a spanning tree $T$ of $G$ and let $F$ be the forest obtained by removing the edges of $T$ that are adjacent to $r$. We call a {\em branch} of $T$ any subgraph consisting of a connected component of $F$, the edge connecting a node of this connected component to $r$ in $T$, and $r$ itself.

For any spanning tree of $G$, each of its branches can belong to one of the following types (see Figure \ref{fig:branches}):
\begin{itemize}
\item Type A: It consists of a single edge connecting the root to a node in $U\cup V$ (see Figure \ref{fig:branches}a).
\item Type B: It consists of an edge connecting the root to a node in $U$ which in turn is connected with its three adjacent nodes of $V$ (see Figure \ref{fig:branches}b).
\item Type C: It consists of an edge connecting the root to a node in $U\cup V$ which is connected to either one or two of its adjacent nodes in $G$ (see Figure \ref{fig:branches}c).
\item Type D: It is a tree of depth exactly $3$ rooted at $r$  (see Figures \ref{fig:branches}d and \ref{fig:branches}e).
\item Type E: It is a tree of depth at least $4$ rooted at $r$ (see Figures \ref{fig:branches}f and \ref{fig:branches}g).
\end{itemize}
\begin{figure}[htbp]
\centerline{\psfig{file=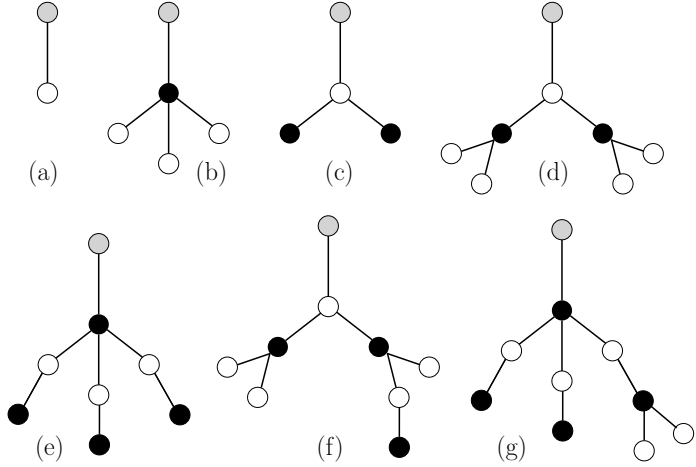,width=3.6in}}  
\caption{Examples with branches considered in the proof of Theorem \ref{thm:pos}. Black and white nodes denote nodes of $U$ and $V$, respectively while the grey nodes represent the root $r$. (a) A branch of type A. (b) A branch of type B. (c) A branch of type C. (d) and (e) Branches of type D. (f) and (g) Branches of type E.} \label{fig:branches}
\end{figure}

We will first prove that if $T$ is an equilibrium for the broadcast game in $G$, then it has a very special structure. In particular, none of its branches rooted at $r$ can be of type $C$, $D$, or $E$. Assume otherwise and let $h$ be such a branch:
\begin{itemize}
\item If $h$ is of type C, consider a leaf $u$ of $h$. The first edge in the path from $u$ to $r$ in $h$ (i.e., the one adjacent to $u$) is not used by any other player besides the one associated with node $u$ while the second edge of the path is used by at most $3$ players (i.e., the players associated with the leaves and the player associated with the node of $h$ which is connected with $r$). Thus, the cost player $u$ experiences is at least $\frac{2+\delta}{3}+1/3>1$ and, hence, this player has an incentive to change her strategy and use the direct edge from $u$ to $r$. See Figure \ref{fig:branches}c for an example.
\item If $h$ is of type D, then it has at most $7$ non-root nodes. Consider a leaf $u$ that is at distance $3$ from $r$. If $u$ belongs to $U$, then its adjacent node in $h$ has degree $2$. Thus, the first edge of the path from $u$ to $r$ in $h$ is not used by any other player besides the one associated with $u$ while the second edge in the path is used by at most $2$ players. In total, the cost the player associated with node $u$ experiences in these two edges of the path is $\frac{2+\delta}{3}+\frac{2+\delta}{6}>1$ and, hence, this player has an incentive to change her strategy and use the direct edge from $u$ to $r$ (see Figure \ref{fig:branches}d). If $u$ belongs to $V$, its adjacent node in $h$ belongs to $U$ and the next node in the path from $u$ to $r$ belongs to $V$. Thus, the first edge of the path from $u$ to $r$ in $h$ is not used by any other player besides the one associated with $u$, the second edge in the path is used by at most $3$ players, and the third edge in the path (the one adjacent to $r$) is used by at most $7$ players. In total, the cost the player associated with node $u$ experiences in these three edges of the path is $\frac{2+\delta}{3}+\frac{2+\delta}{9}+1/7>1$ and, hence, this player has an incentive to change her strategy and use the direct edge from $u$ to $r$ (see Figure \ref{fig:branches}e).
\item If $h$ is of type E, consider a leaf $u$ that is at maximum distance (i.e., at least $4$) from $r$. If $u$ belongs to $U$, then its adjacent node in $h$ has degree $2$. Thus, the first edge of the path from $u$ to $r$ in $h$ is not used by any other player besides the one associated with $u$ while the second edge in the path is used by at most $2$ players. In total, the cost the player associated with node $u$ experiences in these two edges of the path is $\frac{2+\delta}{3}+\frac{2+\delta}{6}>1$ and, hence, this player has an incentive to change her strategy and use the direct edge from $u$ to $r$ (see Figure \ref{fig:branches}f). If $u$ belongs to $V$, its next two nodes in the path from $u$ to $r$ in $h$ belong to $U$ and $V$, respectively. Thus, the first edge of the path from $u$ to $r$ in $h$ is not used by any other player besides the one associated with $u$, the second edge in the path is used by at most $3$ players, and the third edge in the path is used by at most $4$ players. In total, the cost the player associated with node $u$ experiences in these three edges of the path is at least $\frac{2+\delta}{3}+\frac{2+\delta}{9}+\frac{2+\delta}{12}>1$ and, hence, this player has an incentive to change her strategy and use the direct edge from $u$ to $r$ as well (see Figure \ref{fig:branches}g).
\end{itemize}

Instead, if $T$ consists only of branches of types A and B, no player has an incentive to deviate. Indeed, assume that a player associated with a node $u$ in a branch $h_1$ has an incentive to change her strategy and use a new path. Note that the cost she experiences on the edges of $h_1$ she uses is at most $1$. Clearly, her new path cannot include an edge incident to the root which does not belong to any branch since the cost experienced in such a path would be at least $1$. So, assume that the new path of the player contains the edges of another branch $h_2$ in order to connect $u$ to $r$. Clearly, this path should contain an edge of $G$ that is not contained in any branches (and, hence, it is not used by any player besides the one associated with node $u$) while the first edge of branch $h_2$ that the path contains is used by exactly one player besides the one associated with node $u$; this follows by the structure of $G$ and by the fact that branches of $T$ are of type A or B (and, hence, the new path enters branch $h_2$ through one of its leaves). Thus, the cost the player experiences in the new path is at least $\frac{2+\delta}{3}+\frac{2+\delta}{6}>1$ and, hence, she has no incentive to deviate.

Now, consider a spanning tree $T$ that is an equilibrium for the broadcast game in $G$ and let $m$ be the number of branches of type B it contains. Clearly, the weight of the edges in such a branch is $3+\delta$ while the total weight of the edges in branches of type A equals the number of nodes in $U\cup V$ which do not belong to branches of type B, i.e., $5n/2-4m$. Therefore, the total weight of the edges of $T$ is $5n/2-(1-\delta)m$. Let $I$ be the set of nodes of $H$ which correspond to the nodes of the branches of type B that are connected to $r$ in $T$. Due to the structure of $G$ and $T$, $I$ is an independent set of $H$ with size $m$. Also, consider any independent set in $H$ with size $m$. We can conversely construct a spanning tree of $G$ which consists of branches of type A and B and, hence, is an equilibrium: for each node of $U$ corresponding to a node in $I$, we create a branch of type B by connecting this node to $r$ and to its three adjacent nodes in $V$. In this way, we create $m$ branches of type B. Also, we create $5n/2-4m$ branches of type A by connecting each node of $U\cup V$ that does not participate in branches of type B to the root through their direct edges. The cost of this equilibrium tree is again $5n/2-(1-\delta)m$.

Now, we use the inapproximability result due to Berman and Karpinski \cite{BK99}. Their result can be thought of as a polynomial-time reduction from the decision version of {\sc Satisfiability}. The reduction uses a constant $\epsilon\in (0,1/2)$. Given an instance $\phi$ of {\sc Satisfiability}, they construct an instance of {\sc Independent Set} which consists of a $3$-regular graph $H$ with $284k$ nodes (for some parameter $k$) such that
\begin{itemize}
\item $H$ has an independent set of size at least $(140-\epsilon)k$ if $\phi$ is satisfiable, and
\item $H$ has no independent set of size more than $(139+\epsilon)k$ if $\phi$ is not satisfiable.
\end{itemize}

Using the particular graphs as input to our reduction, we can view it as a reduction from {\sc Satisfiability} as well. Given an instance $\phi$ of {\sc Satisfiability}, our reduction defines a broadcast game such that
\begin{itemize}
\item there exists an equilibrium of total weight at most $570+140\delta+(1-\delta)\epsilon$ if $\phi$ is satisfiable, and
\item there exists no equilibrium of total weight less than $571+139\delta-(1-\delta)\epsilon$ if $\phi$ is not satisfiable.
\end{itemize}
By selecting $\epsilon$ and $\delta$ to be arbitrarily small, we conclude that approximating the minimum total weight among all equilibria (and, hence, the price of stability) within a factor better than $571/570$ is NP-hard.
\qed
\end{proof}

\section{Bounds on the amount of subsidies}\label{sec:bounds}
In this section, we provide tight bounds on the amount of subsidies sufficient in order to enforce a minimum spanning tree as an equilibrium in the extension of the original broadcast game. The result is expressed as a constant fraction of the weight of the minimum spanning tree. We first prove our upper bound which is more involved. The proof uses two key ideas: first, the input SNE instance is appropriately decomposed into subinstances of SNE which have a significantly simpler structure. Our decomposition is such that the desired bound has to be  proved for the subinstances; in order to do so, we use a second idea and exploit a virtual cost function that upper-bounds the actual cost experienced by the players in the extension of the game (in the subinstances) with subsidies. The main property of this virtual cost function that simplifies the analysis considerably is that the total amount of subsidies necessary depends only on the weight of the tree (and not on its structure).

\begin{theorem}\label{thm:upper-bound}
Given an instance of {\sc Stable Network Enforcement} consisting of a broadcast game on a graph $G$ and a minimum spanning tree $T$ of $G$, there is a subsidy assignment $b$ of cost at most $\wgt(T)/e$ that enforces $T$ as an equilibrium of the extension of the game with subsidies $b$, where $e$ is the basis of the natural logarithm.
\end{theorem}

\begin{proof}
We decompose the graph $G$ into copies $G^1$, $G^2$, ..., $G^k$ so that the following properties hold:
\begin{itemize}
\item $G^j$ has the same set of nodes and set of edges with $G$.
\item The edge weights in $G^j$ belong to $\{0,c_j\}$ for some $c_j>0$.
\item If the weight of an edge $a$ in $G^j$ is non-zero, then the weight of $a$ is non-zero in each of the copies $G^1, ..., G^{j-1}$ of $G$.
\item The weight of each edge in $G$ is equal to the sum of its weights in the copies of $G$.
\end{itemize}
The decomposition proceeds as follows. Let $c_1$ be the minimum non-zero weight among the edges of $G$. We construct a copy $G^1$ of $G$ (i.e., with the same set of nodes and set of edges) and with edge weights equal to zero if the corresponding edge of $G$ has zero weight and equal to $c_1$ otherwise. Then, we decrease each non-zero edge weight by $c_1$ in $G$ and proceed in the same way with the definition of the edge weights in the copy $G^2$, and so on. We denote by $k$ the number of copies of $G$ that have some edge of non-zero weight. Note that $c_k$ may be infinite if $G$ contains edges of infinite weight, but $k$ is upper bounded by the number of edges in $G$. Clearly, the weight of an edge in the original graph is the sum of its weights in the copies of $G$.

We denote by $T^j$ the spanning tree of $G^j$ that has the same set of edges with $T$. We first observe that $T^j$ is a minimum spanning tree of $G^j$. Assume that this is not the case; then, there must be an edge $a_1$ with zero weight in $G^j$ that does not belong to $T^j$ such that some edge $a_2$ of the edges of $T^j$ with which $a_1$ forms a cycle has non-zero weight $c_j$. By the definition of our decomposition phase, this implies that $a_2$ has higher weight than $a_1$ in $G$. This means that we could remove $a_2$ from $T$ and include $a_1$ in order to obtain a spanning tree with strictly smaller weight, i.e., $T$ would not be a minimum spanning tree.

Now, in order to compute the desired subsidy assignment that enforces $T$ as an equilibrium in the extension of the broadcast game in $G$, we will exploit appropriate subsidy assignments for the broadcast games in each copy of $G$. We have the following lemma.

\begin{lemma}\label{lem:sub-0-c}
Let $c_j>0$. Consider a broadcast game on a graph $G^j$ whose edges have weights in $\{0,c_j\}$ and let $T^j$ be a minimum spanning tree of $G^j$. Then, there is a subsidy assignment $b^j$ of cost at most $\wgt(T^j)/e$ that enforces $T^j$ as an equilibrium in the extension of the game with subsidies $b^j$.
\end{lemma}

\begin{proof}
We call edges of weight $0$ and $c_j$ {\em light} and {\em heavy} edges, respectively. We also call a player associated with a node $v$ a {\em light} player if the weight of the edge connecting $v$ to its parent in $T^j$ is zero; otherwise, we call $v$ a {\em heavy} player. We denote by $m_a$ the number of heavy players which use edge $a$. Clearly, $m_a\leq n_a(T^j)$.

We will introduce a {\em virtual cost} associated with each edge of $T^j$ in order to upper-bound the contribution of the edge to the real cost experienced by each player that uses the edge in $T^j$ in the extension of the game with subsidies. In particular, given subsidies $y_a$ assigned to the heavy edge $a$ with $y_a\in [0,c_j]$, we define the virtual cost of edge $a$ as $\vc(a,y_a) = c_j \ln \frac{m_a}{m_a-1+y_a/c_j}$. The virtual cost of a light edge is always zero; observe that no subsidies have to be assigned to these edges.

\begin{claim}\label{claim:vc-bound}
For any heavy edge $a$ with subsidies $y_a$, it holds that $\vc(a,y_a) \geq \frac{c_j-y_a}{n_a(T^j)}$.
\end{claim}
\begin{proof}
We use the inequality $\ln x \leq x-1$ for $x\in[0,1]$. We have
\begin{eqnarray*}
\vc(a,y_a) &=& c_j \ln{\frac{m_a}{m_a-1+y_a/c_j}}\\
&=& -c_j \ln{\left(1-\frac{c_j-y_a}{m_a c_j}\right)}\\
&\geq&  \frac{c_j-y_a}{m_a}\\
&\geq& \frac{c_j-y_a}{n_a(T^j)}.
\end{eqnarray*}
\qed
\end{proof}

\begin{defn}\label{defn:1}
Consider a path $q$ in $T^j$ and a subsidy assignment $y$ on the edges of $T^j$. We say that $y$ is such that subsidies are packed on the least crowded heavy edges of $q$ if $y_a<c_j$ for a heavy edge $a$ implies that $y_{a'}=0$ for every heavy edge $a'$ of $q$ with $m_{a'}>m_a$.
\end{defn}

We extend the notation of virtual cost so that $\vc(q,b^j)$ denotes the sum of the virtual cost of the edges of a path $q$ in $T^j$ under the subsidy assignment $b^j$. The following claim follows by the definitions and will be very useful later.

\begin{claim}\label{claim:packed}
Consider a path $q$ and denote by $q'$ the set of heavy edges of $q$ and a subsidy assignment $y$. If $\cup_{a\in q'}{\{m_a\}}$ consists of the $|q'|$ consecutive integers $t-|q'|+1, t-|q'|+2, ..., t$, then the virtual cost of path $q$ when subsidies are packed on its least crowded heavy edges is $\vc(q,y) = c_j \ln\frac{t}{t-|q'|+y(q)/c_j}$.
\end{claim}

\begin{proof}
Recall that the only edges that contribute to the virtual cost of $q$ are the heavy edges in $q'$. If $y(q)=0$ (i.e., no subsidies are put on the edges of $q'$), the virtual cost is
\begin{eqnarray*}
\vc(q,y) &=& \sum_{a\in q'}{\vc(a,y_a)}\\
&=& \sum_{a\in q'}{c_j \ln{\frac{m_{a}}{m_{a}-1+y_a/c_j}}}\\
&=&\sum_{i=t-|q'|+1}^t{c_j\ln{\frac{i}{i-1}}}\\
&=&c_j\ln{\frac{t}{t-|q'|+y(q)/c_j}}.
\end{eqnarray*}
The first two equalities follow by the definition of the virtual cost, the third one follows since $\cup_{a\in q'}{\{m_a\}}=\{t-|q'|+1, t-|q'|+2, ..., t\}$ and $y_a=0$, and the last one is obvious.

We now consider the case $y(q)>0$. Since subsidies are packed on the least crowded heavy edges of $q$, there must be a heavy edge $a\in q'$ such that $y_a>0$ so that $y_{a'}=0$ for each heavy edge $a'$ with $m_{a'}>m_a$ and $y_{a''}=c_j$ for each heavy edge $a''$ with $m_{a''}<m_a$. Let $q'_1=\{a'\in q':y_{a'}=0\}$ and $q'_2=q'\setminus (q'_1\cup \{a\})$. Observe that the edges of $q'_2$ and the light edges of $q$ do not contribute to the virtual cost of $q$. Hence,
\begin{eqnarray*}
\vc(q,y) &=& \sum_{a'\in q'_1}{\vc(a',y_{a'})} +\vc(a,y_a) \\
&=& \sum_{a'\in q'_1}{c_j\ln{\frac{m_{a'}}{m_{a'}-1}}}+c_j\ln{\frac{m_a}{m_a-1+y_a/c_j}}\\
&=& c_j\sum_{i=t-|q'_1|+1}^t{\ln{\frac{i}{i-1}}}+c_j\ln\frac{t-|q'_1|}{t-|q'_1|-1+y_a/c_j}\\
&=& c_j\ln{\frac{t}{t-|q'_1|-1+y_a/c_j}}\\
&=& c_j\ln{\frac{t}{t-|q'|+y(q)/c_j}}
\end{eqnarray*}
The first two equalities follow by the definition of the virtual cost, the third one follows since the definition of $q'_1$ implies that $\cup_{a'\in q'_1}{\{m_{a'}\}}=\{t-|q'_1|+1, t-|q'_1|+2, ..., t\}$ and $m_a=t-|q'_1|$, the fourth equality is obvious, and the last one follows since $y(q)=y_a+|q'_2|c_j$ and $|q'|=|q'_1|+|q'_2|+1$.
\qed\end{proof}

Figure \ref{fig:visual} provides a visualization of the virtual cost in a path when subsidies are packed on its less crowded heavy edges.

\begin{figure}[htbp]
\centerline{\psfig{file=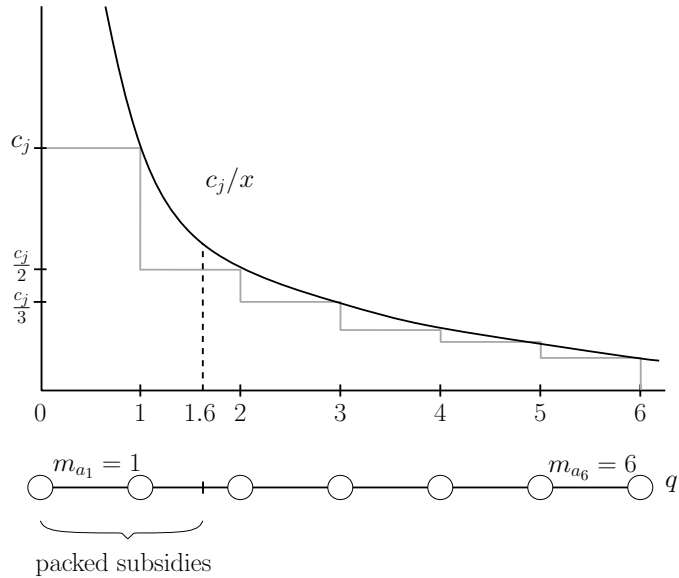,width=3.5in, height=3in}}  
\caption{A visualization of the virtual cost in a path $q$, with $6$ heavy edges and $\cup_{a\in q}{\{m_a\}}=\{1, 2, ..., 6\}$, when subsidies are packed on its less crowded edges. The leftmost edge and a fraction of 60\% of the second leftmost one have been subsidized. The virtual cost $\ln\frac{6}{1.6}$ (see Claim \ref{claim:packed}) is the area to the right of the dashed line that is below the black line. The real cost experienced by the player associated with the far left node is the area to the right of the dashed line that is below the grey line.
} \label{fig:visual}
\end{figure}

Now, we compute the subsidy assignment $b^j$ that assigns no subsidies to the light edges and subsidies to the heavy edges of $T^j$ as follows. Denote by $L$ the set of leaf-nodes of $T^j$ such that the path $T^j_u$ connecting such a leaf-node $u$ to the root node $r$ in $T^j$ contains at least one heavy edge. For each leaf-node $u$ of $L$, we pack subsidies to the least crowded heavy edges of $T^j_u$ so that the virtual cost on the path $T^j_u$ is exactly $c_j$. In particular, let $S$ be the set of edges of $T^j$ defined as follows: a heavy edge $(v,p(v))$ belongs to $S$ if $\vc(T^j_{p(v)},0)< c_j$ and $\vc(T^j_v,0)\geq c_j$. Observe that the set $S$ disconnects the leaves of $L$ from the root node. Indeed, if this was not the case, there would be a heavy edge that is used by exactly one heavy player and is not assigned any subsidies; by the definition of the virtual cost, its virtual cost would be infinite. All heavy edges that are on the side of the partition together with the root node are assigned zero subsidies; the heavy edges on the other side of the partition are assigned subsidies of $c_j$ and do not contribute to the virtual cost of the paths they belong to. An edge $a=(v,p(v))$ of $S$ is assigned subsidies $b^j_a $ with
$$b^j_a = c_j\left(1-m_a \left(1-\exp\left(\frac{\vc(T^j_{p(v)},0)}{c_j}-1\right)\right)\right).$$
This definition implies that $\vc(T^j_{p(v)},0)+\vc(a,b^j_a) = c_j$. In this way, we guarantee that the virtual cost of any path to the root node $r$ is at most $c_j$ if it contains at least one heavy edge and zero otherwise.

We will now show that, given the subsidies we have assigned to the edges of $T^j$, no player has an incentive to deviate from her path to the root in $T^j$. Consider the player associated with a node $u$ and let $q_u$ be another path from $u$ to $r$ in $T^j$. Recall that the definition of subsidies and Claim \ref{claim:vc-bound} guarantees that the cost experienced by the player that uses $T^j_u$ is at most $c_j$. Now, consider the edges of $q_u$ that do not belong to $T$ (since $q_u\not=T^j_u$ there is at least one such edge). If any such edge has weight $c_j$, this means that, by deviating to $q_u$, the player associated with node $u$ would experience a cost of at least $c_j$ and, hence, has no incentive to do so. So, in the following we assume that the edges of $q_u$ that do not belong to $T^j$ have zero weight. Now, consider the subgraph $H$ of $G^j$ induced by the edges in the paths $T^j_u$ and $q_u$. Let $C$ be a cycle of $H$. It consists of edges of $T^j$ and edges not belonging to $T^j$ that have zero weight. This implies that all edges in $C$ have zero weight, otherwise we could replace an edge of $C$ that belongs to $T^j$ (and has non-zero weight) with an edge of $C$ that does not belong to $T$ (which has zero weight) and obtain a spanning tree of $G^j$ with strictly smaller weight than $T^j$. This contradicts the assumption that $T^j$ is a minimum spanning tree of $G^j$. So, all edges of $T^j_u$ that are contained in a cycle in $H$ have zero weight. The remaining edges of $T^j_u$ are also used by $q_u$. We conclude that the total cost experienced by the player associated with node $u$ is the same no matter whether she uses path $T^j_u$ or $q_u$ and, hence, she has no incentive to deviate from path $T^j_u$ to $q_u$.

We will now show that the total amount of subsidies put on the edges of $T^j$ in this way is exactly $\wgt(T^j)/e$. In order to show this, we will show that the total amount of subsidies put on the edges of $T^j$ equals the total amount of subsidies put by the same procedure on the edges of another tree $\bar{T}$ that consists of a single path from the root that spans all the nodes and has the same number of heavy edges as the original one. As an intermediate step, consider two edges $g_1=(v_1, p(v_1))$ and $g_2=(v_2,p(v_2))$ of $S$ such that the least common ancestor $u$ of nodes $v_1$ and $v_2$ in $T^j$ has largest depth. We denote by $h_1$ and $h_2$ the number of heavy edges in the subtrees of $v_1$ and $v_2$, respectively, and by $q_1$ and $q_2$ the paths connecting $u$ to $v_1$ and $v_2$ in $T$, respectively. Also, denote by $q'_1$ (resp. $q'_2$) the subset of $q_1$ (resp. $q_2$) consisting of heavy edges. Assume that the virtual cost of the path in $T^j$ from $r$ to $u$ is $\ell$ for some $\ell\in [0,c_j)$; then, the virtual cost of the paths $q_1$ and $q_2$ is exactly $c_j-\ell$. Denote by $b^j_{g_1}$ and $b^j_{g_2}$ the subsidies assigned to edges $g_1$ and $g_2$ by the above procedure, respectively. Since the edges $g_1$ and $g_2$ are selected so that their least common ancestor $u$ has largest depth, the edges in the path $q_1$ are not used by any heavy player different than those in the subtree of $T^j$ rooted at $v_1$. Similarly, the edges in the path $q_2$ are not used by any heavy player other than those in the subtree of $T^j$ rooted at $v_2$. Hence, both paths $q_1$ and $q_2$ satisfy the condition of Claim \ref{claim:packed} above in the sense that $\cup_{a\in q'_1}{\{m_a\}} = \{h_1+1, h_1+2, ..., h_1+|q'_1|\}$ and $\cup_{a\in q'_2}{\{m_a\}} = \{h_2+1, h_2+2, ..., h_2+|q'_2|\}$, respectively. Hence, we can express the virtual cost of paths $q_1$ and $q_2$, respectively, as
$$\vc(q_1,b^j)=c_j\ln \frac{h_1+|q'_1|}{h_1+b^j_{g_1}/c_j} = c_j-\ell\,$$
and
$$\vc(q_2,b^j)=c_j\ln \frac{h_2+|q'_2|}{h_2+b^j_{g_2}/c_j} = c_j-\ell.$$
Equivalently, we have $h_1+|q'_1| = \exp(1-\ell/c_j) (h_1+b^j_{q_1}/c_j)$ and $h_2+|q'_2| = \exp(1-\ell/c_j) (h_2+b^j_{g_2}/c_j)$.
By summing these last two equalities, we obtain that $h_1+h_2+|q'_1|+|q'_2| = \exp(1-\ell/c_j)(h_1+h_2+(b^j_{g_1}+b^j_{g_2})/c_j)$ which implies
\begin{eqnarray}\label{eq:vc1}
c_j\ln\frac{h_1+h_2+|q'_1|+|q'_2|}{h_1+h_2+(b^j_{g_1}+b^j_{g_2})/c_j} &=& c_j-\ell.
\end{eqnarray}

Now, consider the following transformation of $T^j$ to another tree $T'$. The only change is performed in the paths $q_1$, $q_2$ and all subtrees of their nodes besides node $u$. We replace all these edges in $T^j$ with a path $q$ originating from $u$ and spanning all the nodes in $q_1$ and $q_2$ and their subtrees so that exactly $h_1+h_2+|q'_1|+|q'_2|$ heavy edges are used. Let $q'$ be the set of heavy edges in $q$. We pack a total amount $c_j(h_1+h_2)+b^j_{g_1}+b^j_{g_2}$ of subsidies (i.e., the same total amount of subsidies used in the heavy edges of $q_1$, $q_2$ and in the subtrees of nodes $v_1$ and $v_2$ in $T^j$) on the least crowded heavy edges of path $q$ while the assignment of subsidies on the other heavy edges of $T'$ is the same as in the corresponding edges in $T^j$. Now, the path $q$ satisfies the condition of the Claim \ref{claim:packed} above in the sense that $\cup_{a\in q'}{\{m_a\}}=\{1, ..., h_1+h_2+|q'_1|+|q'_2|\}$. Hence, the virtual cost of path $q$ when a total amount $c_j(h_1+h_2)+b^j_{g_1}+b^j_{g_2}$ of subsidies is packed on its least crowded heavy edges is the one at the left hand side of equality (\ref{eq:vc1}) and is exactly $c_j-\ell$ while the virtual cost of the path from the root to $u$ in $T'$ is not affected by our transformation (the number of heavy players in the subtree of node $u$ stays the same after the transformation) and is equal to $\ell$. Hence, we have transformed $T^j$ to $T'$ so that the same total amount of subsidies is used and guarantees that any path from the root to a node has virtual cost at most $c_j$. By executing the same transformation in $T'$ repeatedly, we end up with a tree $\bar{T}$ which consists of a path $\bar{q}$ spanning all the nodes and has the same number of heavy edges as the original tree $T^j$ (and, obviously, the same total weight). Let $\bar{q}'$ be the set of heavy edges in $\bar{q}$. The transformation guarantees that by packing the original total amount of subsidies on the least crowded heavy edges of $\bar{q}$, we have that its virtual cost is exactly $c_j$. Also, note that $\cup_{a\in \bar{q}'}{\{m_a\}} = \{1, 2, ..., |\bar{q}'|\}$ and, by Claim \ref{claim:packed}, the virtual cost of path $\bar{q}$ when a total amount $b(T^j)$ of subsidies is packed on its least crowded heavy edges is $c_j\ln\frac{|\bar{q}'|c_j}{b(T^j)}=c_j$. This implies that the total amount of subsidies in the original tree is $b(T^j) = |\bar{q}'| c_j/e = \wgt(T^j)/e$.
\qed\end{proof}

Now, for each copy $G^j$ of $G$, we use the procedure in the proof of Lemma \ref{lem:sub-0-c} to compute a subsidy assignment $b^j$ so that the tree $T^j$ is an equilibrium for the extension of the broadcast game on the graph $G^j$ with subsidies $b^j$. For the original game on the graph $G$, we assign an amount of $b_a = \sum_{j=1}^{k}{b_a^j}$ as subsidies to edge $a$ (i.e., equal to the total amount of subsidies assigned to $a$ for each copy of $G$). We can easily show that $T$ is an equilibrium for the original broadcast game. Let $T_u$ be the path used by the player associated with node $u$ in $T$ and denote by $q_u$ a different path connecting $u$ with $r$ in $G$. The cost experienced by the player associated with node $u$ in $T$ is
\begin{eqnarray*}
\co_u(T;b) &=& \sum_{a\in T_u}{\frac{w_a-b_a}{n_a(T)}}\\
&=& \sum_{a\in T_u}\sum_{j=1}^k{\frac{w_a^j-b_a^j}{n_a(T^j)}}\\
&=& \sum_{j=1}^k\sum_{a\in T_u}{\frac{w_a^j-b_a^j}{n_a(T^j)}}\\
&\leq & \sum_{j=1}^k\sum_{a\in q_u}{\frac{w_a^j-b_a^j}{n_a(T^j)+1-n_a^u(T^j)}}\\
&=& \sum_{a\in q_u}{\frac{w_a-b_a}{n_a(T_{-u},q'_u)}} \\
&=& \co_u(T_{-u},q_u;b),
\end{eqnarray*}
i.e., not larger than the cost she would experience by deviating to path $q_u$, which implies that $T$ is indeed enforced as an equilibrium in the extension of the broadcast game on $G$ with the particular subsidies. The equalities follow by the definition of the cost experienced by the player associated with node $u$, or the definition of our decomposition, or due to the exchange of sums, and the inequality follows since, by Lemma \ref{lem:sub-0-c}, $T^j$ is enforced as an equilibrium in the extension of the broadcast game in $G^j$ with subsidies $b^j$. The bound on the amount of subsidies follows by the guarantee for the cost of the subsidy assignments $b^j$ from Lemma \ref{lem:sub-0-c} and the last property of our decomposition, and since the subsidy $b_a$ on each edge $a$ is defined as $b_a=\sum_{j=1}^k{b^j_a}$.
\qed\end{proof}

We now present our lower bound.

\begin{theorem}\label{thm:lb}
For every $\epsilon>0$, there exist a broadcast game on a graph $G$ and a minimum spanning tree $T$ of $G$ such that any subsidy assignment that enforces $T$ as equilibrium of the extension of the broadcast game with subsidies is at least $(1/e-\epsilon)\wgt(T)$.
\end{theorem}

\begin{proof}
Consider the graph $G$ which consists of a cycle with $n+1$ edges of unit weight that span the root node $r$ and the $n$ nodes which are associated with the players. Let $a=(r,u)$ be an edge incident to the root node $r$ and let $T$ be the path that contains all edges of $G$ besides $a$. Clearly, $T$ is a minimum spanning tree of $G$. Now, in order to satisfy that the player associated with node $u$ has no incentive to deviate from her strategy in $T$ and use edge $a$ instead, we have to put subsidies on some of the edges of the path $T$. The maximum decrease in the cost of the player associated with node $u$ is obtained when subsidies are packed on the least crowded edges of $T$ (i.e., on the edges of $T$ that are further from the root); equivalently, the minimum amount of subsidies necessary in order to decrease the cost of this player to $1$ is obtained when subsidies are packed on the least crowded edges of $T$. Let $k$ be the number of edges that are subsidized. Since the player associated with node $u$ has no incentive to deviate, the cost of $\H_n-\H_{k}$ she experiences at the $n-k$ edges on which we do not put subsidies is at most $1$ while the total amount of subsidies is at least $k-1$. Using the inequality $x\geq \ln{(1+x)}$ for $x\geq 0$, we obtain $1\geq \H_n-\H_k = \sum_{t=k+1}^n{\frac{1}{t}}\geq \sum_{t=k+1}^n{\ln\frac{t+1}{t}} = \ln\frac{n+1}{k+1}$ which implies that that the total amount of subsidies is at least $k-1\geq \frac{n+1}{e}-2$. The weight of $T$ is $n$ and the bound follows by selecting $n$ to be sufficiently large.
\qed\end{proof}

\section{All-or-nothing subsidies}\label{sec:all-or-nothing}
In this section, we consider the all-or-nothing version of SNE. Interestingly, in contrast to the standard version, we prove (Theorem \ref{thm:hardness-all-or-nothing}) that its optimization version is hard to approximate within any factor.

\begin{theorem}\label{thm:hardness-all-or-nothing}
Given an instance of all-or-nothing {\sc Stable Network Enforcement} consisting of a broadcast game on a graph $G$ and a minimum spanning tree $T$ of $G$, approximating (within any factor) the minimum cost over all-or-nothing subsidy assignments that enforce $T$ as an equilibrium in the extension of the broadcast game is NP-hard.
\end{theorem}

\begin{proof}
We use a reduction from instances of 3SAT-4 which consist of CNF formulas such that each clause contains exactly three literals (corresponding to different variables) and each variable appears in at most four clauses. Deciding whether such an instance has a truthful assignment or not is NP-hard \cite{T84}.
Given an instance $\phi$ of 3SAT-4 with a set of clauses $C$, our construction defines (in polynomial time) a broadcast game on a graph $G$ and identifies a particular minimum spanning tree $T$ of $G$. Our construction uses a parameter $K$ which is significantly larger than $3|C|$. We will show that deciding whether the minimum cost over all-or-nothing subsidy assignments that enforce $T$ as an equilibrium in the extension of the broadcast game is at most $3|C|$ or at least $K$ is NP-hard.

First, we assign a {\em label} from $\{1, 2, ..., 9\}$ to each variable in such a way that two variables that appear in the same clause are assigned different labels. Due to the fact that each variable appears in at most four clauses, it should have a different label than at most eight other variables. Hence, nine labels suffice and the corresponding labeling can be computed in polynomial time. For $j=1, 2, ..., 9$, we define the constants $n_j=\frac{1}{4} 28^{2^{9-j}}$ or, equivalently, for $j=1, 2, ..., 8$, $n_j=4n_{j+1}^2$ with $n_9=7$.

The graph $G$ has a root-node $r$ and consists of several gadgets: {\em literal gadgets}, {\em clause gadgets}, {\em consistency gadgets}, and {\em auxiliary} nodes and edges. The non-root nodes of $G$ are partitioned into {\em critical} and {\em non-critical} ones. The edges of $G$ belong to three different types: {\em heavy edges} of weight at least $K$, {\em ultra light edges} of zero weight, and {\em light edges} of unit weight.

We start with the definition of the literal gadgets (see Figure \ref{fig:literal}). For each appearance of a literal $\ell$ in a clause $c$, we have a literal gadget which consists of four non-critical nodes $l(c,\ell)$, $u(c,\bar{\ell})$, $u(c,\ell)$ and $v_1(c,\ell)$, and the critical nodes $v_2(c,\ell)$ and $v_3(c,\ell)$. Let $j$ denote the label of the variable corresponding to literal $\ell$. Then, there are the following edges: the light edges $(l(c,\ell),u(c,\bar{\ell}))$ and $(u(c,\bar{\ell}),u(c,\ell))$, the heavy edges $(l(c,\ell),v_1(c,\ell))$, $(v_1(c,\ell),v_2(c,\ell))$, and $(v_3(c,\ell),u(c,\ell))$ of weight $K$, the heavy edge $(l(c,\ell),v_3(c,\ell))$ of weight $K+\frac{1}{n_j-3}$ and the heavy edge $(v_2(c,\ell),u(c,\ell))$ of weight $\frac{3K}{2}-\frac{1}{n_j+1}$. Among them, the first five edges belong to $T$ while the last two ones do not.

\begin{figure}[htbp]
\centerline{\psfig{file=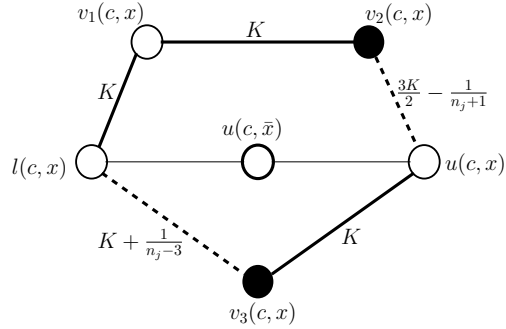,height=1.7in}}  
\caption{The literal gadget for the appearance of literal $x$ in clause $c$. The black nodes are the critical ones and the solid edges are the ones that belong to $T$. The thick and thin solid edges are the heavy and light ones, respectively.} \label{fig:literal}
\end{figure}

For each clause $c=(\ell_1,\ell_2,\ell_3)$ with literals corresponding to different variables with labels $j_1$, $j_2$, and $j_3$ with $j_1<j_2<j_3$, we have a clause gadget (see Figure \ref{fig:clause}) which connects the three literal gadgets corresponding to the appearance of literals $\ell_1$, $\ell_2$, and $\ell_3$ in clause $c$ as follows. The node $l(c,\ell_1)$ coincides with the root-node $r$, the node $l(c,\ell_2)$ coincides with node $u(c,\ell_1)$, and the node $l(c,\ell_3)$ coincides with $u(c,\ell_2)$. There is an extra critical node $v(c)$ which is connected through a heavy edge of weight $K$ to node $u(c,\ell_3)$ and through a heavy edge of weight $K+\frac{1}{n_{j_1}}+\frac{1}{n_{j_2}-3}+\frac{1}{n_{j_3}-3}$ to the root-node $r$. Among these two edges, the first one belongs to $T$ while the second one does not.

\begin{figure*}[htbp]
\centerline{\psfig{file=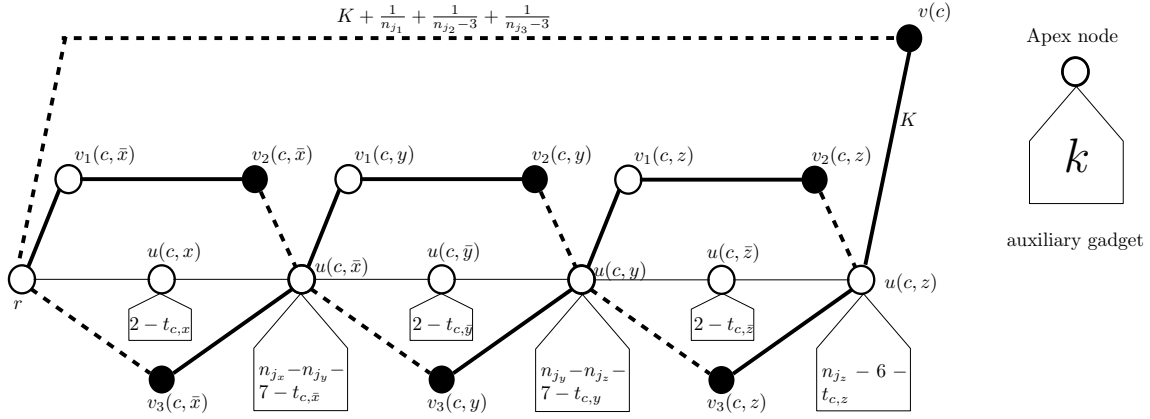,width=6in}}  
\caption{The clause gadget for clause $c=(\bar{x},y,z)$ along with the auxiliary nodes and edges. The auxiliary gadget (shown in the right) indicates that $k$ auxiliary nodes are connected to the apex node via ultra light edges} \label{fig:clause}.
\end{figure*}

Let $c_1$, $c_2$, $c_3$, and $c_4$ be the (at most) four clauses in which variable $x$ appears. For each $i=1, 2, 3$, we have a consistency gadget that connects the two literal gadgets corresponding to the appearance of variable $x$ or its negation in clauses $c_i$ and $c_{i+1}$. Let $j$ be the label of variable $x$. We use two different consistency gadgets depending on whether the variable appears as the same literal in both $c_i$ and $c_{i+1}$ or not. An $\ell$-$\ell$ consistency gadget corresponds to the appearance of literal $\ell$ in clauses $c_i$ and $c_{i+1}$ and consists of two critical nodes $u_1(c_i,c_{i+1},\ell)$ and $u_2(c_i,c_{i+1},\ell)$. Node $u_1(c_i,c_{i+1},\ell)$ is connected through a heavy edge of weight $K$ to node $u(c_i,\bar{\ell})$ of the literal gadget corresponding to the appearance of the literal $\ell$ in $c_i$ and a heavy edge of weight $K+\frac{1}{2n_j}$ to node $u(c_{i+1},\bar{\ell})$ of the literal gadget corresponding to the appearance of literal $\ell$ in $c_{i+1}$. Among these two edges, the first one belongs to $T$ while the second one does not.
Node $u_2(c_i,c_{i+1},\ell)$ is connected through a heavy edge of weight $K$ to node $u(c_{i+1},\bar{\ell})$ of the literal gadget corresponding to the appearance of the literal $\ell$ in $c_{i+1}$ and a heavy edge of weight $K+\frac{1}{2n_j}$ to node $u(c_{i},\bar{\ell})$ of the literal gadget corresponding to the appearance of literal $\ell$ in $c_{i}$. Among these two edges, the first one belongs to $T$ while the second one does not.
An $\ell$-$\bar{\ell}$ consistency gadget corresponds to the appearance of literals $\ell$ and $\bar{\ell}$ in clauses $c_i$ and $c_{i+1}$, respectively, and consists of two critical nodes $u_1(c_i,c_{i+1},\ell)$ and $u_2(c_i,c_{i+1},\ell)$. Node $u_1(c_i,c_{i+1},\ell)$ is connected through a heavy edge of weight $K$ to node $u(c_i,\ell)$ of the literal gadget corresponding to the appearance of the literal $\ell$ in $c_i$ and a heavy edge of weight $K+\frac{1}{n_j}+\frac{1}{2n_j^2}$ to node $u(c_{i+1},\ell)$ of the literal gadget corresponding to the appearance of literal $\bar{\ell}$ in $c_{i+1}$. Among these two edges, the first one belongs to $T$ while the second one does not. Node $u_2(c_i,c_{i+1},\ell)$ is connected through a heavy edge of weight $K$ to node $u(c_{i+1},\ell)$ of the literal gadget corresponding to the appearance of the literal $\bar{\ell}$ in $c_{i+1}$ and a heavy edge of weight $K$ to node $u(c_{i},\ell)$ of the literal gadget corresponding to the appearance of literal $\ell$ in $c_{i}$. Among these two edges, the first one belongs to $T$ while the second one does not. An example is depicted in Figure \ref{fig:consistency}.

\begin{figure}[ht]
\centerline{\psfig{file=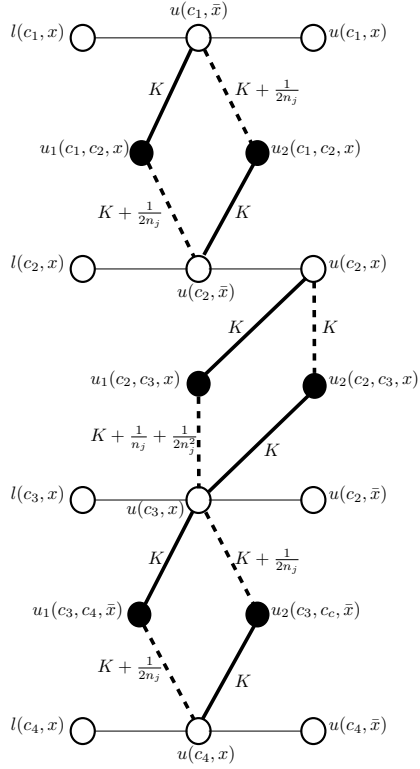,height=4in}}  
\caption{Three consistency gadgets connecting literal gadgets associated to the appearance of literal $x$ in clause $c_1$, literal $x$ in clause $c_2$, literal $\bar{x}$ in clause $c_3$, and literal $\bar{x}$ in clause $c_4$. The first and the third ones are $\ell$-$\ell$ consistency gadgets; the second one is an $\ell$-$\bar{\ell}$ consistency gadget. Only nodes $l(c,\ell)$, $u(c,\bar{\ell})$ and $u(c,\ell)$ are shown in each literal gadget. } \label{fig:consistency}
\end{figure}

The last step in the construction is to include auxiliary non-critical nodes connected through ultra light edges to nodes $u(c,\bar{\ell})$ and $u(c,\ell)$ for each clause $c$ and each literal $\ell$ that appear in $c$. These edges belong to $T$. The purpose of the auxiliary nodes and ultra light edges is to guarantee that for each appearance of a variable of label $j$ or its negation as a literal $\ell$ in a clause $c$, the number of players in $T$ that use edges $(l(c,\ell),u(c,\bar{\ell}))$ and $(u(c,\bar{\ell}),u(c,\ell))$ is exactly $n_j$ and $n_j-3$, respectively. This is done as follows. Consider a clause $c=(\ell_1,\ell_2,\ell_3)$ such that the literals $\ell_1$, $\ell_2$, and $\ell_3$ correspond to variables with labels $j_1$, $j_2$, and $j_3$ with $j_1<j_2<j_3$. For $i=1,2,3$, let $t_{c,\bar{\ell}_i}$ be the number of nodes in consistency gadgets to which node $u(c,\bar{\ell_i})$ is connected in $T$ and observe that $t_{c,\bar{\ell}_i}\in \{0,1,2\}$. We connect it to $2-t_{c,\bar{\ell}_i}$ additional auxiliary nodes through ultra light edges. Node $u(c,\ell_3)$ is connected in $T$ to node $v(c)$, node $v_3(c,\ell_3)$, and to $t_{c,\ell_3}$ nodes in the consistency gadgets it participates with $t_{c,\ell_3}\in \{0,1\}$. We connect $u(c,\ell_3)$ to $n_{j_3}-6-t_{c,\ell_3}$ additional auxiliary nodes through ultra light edges. For $i=1,2$, node $u(c,\ell_i)$ is connected in $T$ to nodes $v_3(c,\ell_i)$, $v_1(c,\ell_{i+1})$, and to $t_{c,\ell_i}$ nodes in the consistency gadgets it participates with $t_{c,\ell_i}\in \{0,1\}$; we connect it to $n_{j_i}-n_{j_{i+1}}-7-t_{c,\ell_i}$ additional nodes through ultra light edges (see Figure \ref{fig:clause}).

In the following, we refer to players associated to non-critical nodes as {\em non-critical players}. All other players are {\em critical}. The critical players associated with nodes $v_2(c,\ell)$ and $v_3(c,\ell)$ of a literal gadget corresponding to the appearance of literal $\ell$ in clause $c$ are called {\em literal players}. The critical players associated with nodes $u_1(c_1,c_2,\ell)$ and $u_2(c_1,c_2,\ell)$ in consistency gadgets corresponding to the appearance of literal $\ell$ (or literals $\ell$ and $\bar{\ell}$) in clauses $c_1$ and $c_2$ are called {\em consistency players}. The critical player associated with node $v(c)$ in the clause gadget corresponding to clause $c$ is called {\em clause player}.

Observe that our construction guarantees that $T$ is connected, spans the nodes of $G$, and the number of its edges is equal to the number of non-root nodes in $G$. Hence, it is indeed a tree. Also, note that it consists of all light and ultra light edges and heavy edges of weight exactly $K$; the edges of $G$ not included in $T$ are all heavy and, hence, $T$ is a minimum spanning tree.

We have completed the description of our reduction. We use the term {\em light assignment} to refer to all-or-nothing subsidy assignments that subsidize only light edges (clearly, ultra light edges do not need to be subsidized). In the rest of the proof, we show that there exists a light assignment that enforces $T$ as an equilibrium in the extension of the broadcast game on $G$ if and only if $\phi$ is satisfiable (Corollary \ref{cor:T-sat-eq}). This is done in a sequence of steps which can be briefly described as follows:
\begin{itemize}
\item Step 1. First, we observe that light assignments of subsidies guarantee that the non-critical players have no incentive to deviate from their strategy in $T$ (Lemma \ref{lem:non-critical}).
\item Step 2. Then, we introduce the property of {\em balance} for light assignments which is proved to be equivalent to the fact that the critical literal players do not have an incentive to deviate either; this is done in Lemma \ref{lem:literal-gadget} using the definition of the literal gadgets.
\item Step 3. Then, we introduce the property of {\em consistency} for balanced light assignments which is proved to be equivalent to the fact that the critical consistency players do not have an incentive to deviate either; this is done in Corollary \ref{cor:consistency} using the definition of the consistency gadgets. A nice effect of this property is that there is a one-to-one and onto mapping between the consistent balanced light assignments and the assignments of values to the variables of $\phi$.
\item Step 4. Finally, we introduce an additional property for consistent balanced light assignments that is proved to be equivalent to the fact that the critical clause players do not have an incentive to deviate either and, hence, $T$ is enforced as an equilibrium in the broadcast game; this is done in Lemma \ref{lem:clause-gadget} using the definition of the clause gadgets. A nice effect of this property is that the mapping mentioned above is a one-to-one and onto mapping between the consistent balanced light assignments that satisfy this property and the truthful assignments of $\phi$.
\end{itemize}

We continue with Step 1.

\begin{lemma}\label{lem:non-critical}
Consider the extension of the broadcast game on $G$ with a light assignment of subsidies. Then, no non-critical player has an incentive to change her strategy in $T$.
\end{lemma}

\begin{proof}
Consider a non-critical player whose strategy in $T$ consists of light (and, possibly, ultra light) edges. By the construction of $T$, this player may use light edges in at most three literal gadgets. Hence, the cost she experiences is at most $6$ while any deviation should include a heavy edge out of $T$.

Now consider the non-critical player associated with node $v_1(c,\ell)$ in the literal gadget corresponding to the appearance of literal $\ell$ in clause $c$. In her strategy in $T$, she uses the edge $(v_1(c,\ell),l(c,\ell))$ (which is also used by the player associated with node $v_2(c,\ell)$) as well as at most four light edges in at most two literal gadgets. Hence, her cost is at most $K/2+4$. If she deviates to a strategy that contains edge $(v_1(c,\ell),l(c,\ell))$ but not the path from $l(c,\ell)$ to $r$ in $T$, she would experience a cost of at least $3K/2>K/2+4$ since, in addition to edge $(v_1(c,\ell),l(c,\ell))$, her strategy would include another heavy edge out of $T$ used only by her. Also, any deviation that contains the path $\langle v_1(c,\ell),v_2(c,\ell),u(c,\ell)\rangle$ would have cost at least $2K-\frac{1}{n_j+1}>K/2+6$, where $j$ is the label of the variable corresponding to literal $\ell$.
\qed\end{proof}

We proceed with Step 2. Consider a light assignment such that, for each clause $c$ and each literal $\ell$ that appears in $c$, exactly one among the light edges $(l(c,\ell),u(c,\bar{\ell}))$ and $(u(c,\bar{\ell}),u(c,\ell))$ in the literal gadget corresponding to the appearance of literal $\ell$ in clause $c$ is subsidized; we call such an assignment a {\em balanced} light assignment.

\begin{lemma}\label{lem:literal-gadget}
Consider the extension of the broadcast game on $G$ with a light assignment of subsidies. Then, the critical literal players have no incentive to change their strategies in $T$ if and only if the assignment is balanced light.
\end{lemma}

\begin{proof}
Consider the literal gadget that corresponds to the appearance of a literal $\ell$ in a clause $c$. If none of the two light edges of the gadget is subsidized, the critical literal player associated with node $v_3(c,\ell)$ has an incentive to deviate from her path $\langle v_3(c,\ell), u(c,\ell), u(c,\bar{\ell}), l(c,\ell)\rangle$ (where the cost she experiences is $K+\frac{1}{n_j-3}+\frac{1}{n_j}$) to the direct edge from $v_3(c,\ell)$ to $l(c,\ell)$ (where the cost she would experience would be $K+\frac{1}{n_j-3}$). If both of the two light edges of the gadget are subsidized, the critical literal player associated with node $v_2(c,\ell)$ has an incentive to deviate from her path $\langle v_2(c,\ell), v_1(c,\ell), l(c,\ell)\rangle$ (where the cost she experiences is $3K/2$) to the path $\langle v_2(c,\ell), u(c,\ell), u(c,\bar{\ell}), l(c,\ell)\rangle$ (where the cost she would experience would be $\frac{3K}{2}-\frac{1}{n_j+1}$). If instead one of the edges $(l(c,\ell),u(c,\bar{\ell}))$ and $(u(c,\bar{\ell}), u(c,\ell)$ is subsidized, the critical literal player associated with node $v_2(c,\ell)$ experiences a cost of $3K/2$ and has no incentive to change her strategy to the path $\langle v_2(c,\ell), u(c,\ell), u(c,\bar{\ell}), l(c,\ell)\rangle$ (since her cost there would be at least $\frac{3K}{2}-\frac{1}{n_j+1}+\min\{\frac{1}{n_j+1},\frac{1}{n_j-2}\}\geq \frac{3K}{2})$. Also, the critical literal player associated with node $v_3(c,\ell)$ experiences a cost of at most $K+\max\{\frac{1}{n_j-3}, \frac{1}{n_j}\}$ and has no incentive to change her path from $\langle v_2(c,\ell), u(c,\ell), u(c,\bar{\ell}), l(c,\ell)\rangle$ to the direct edge from $v_3(c,\ell)$ to $l(c,\ell)$ (since her cost there would be $K+\frac{1}{n_j-3}\geq K+\max\{\frac{1}{n_j-3}, \frac{1}{n_j}\}$). Also, note that if the critical literal player associated with node $v_2(c,\ell)$ or $v_3(c,\ell)$ deviates to a strategy that does not use the path from node $l(c,\ell)$ to the root-node $r$ in $T$, this would contain two heavy edges which are used only by this player for a cost of at least $2K$.
\qed\end{proof}

We proceed with Step 3 where our goal is to prove Corollary \ref{cor:consistency}. This will follow by Lemmas \ref{lem:ell-ell-consistency-gadget} and \ref{lem:ell-not-ell-consistency-gadget}. In their proof, we will use the following additional lemma.

\begin{lemma}\label{lem:distance}
Consider the extension of the broadcast game on $G$ with a balanced light assignment and a literal gadget corresponding to the appearance of literal $\ell$ in clause $c$. Then, any player who uses or deviates to a strategy that contains the path from node $l(c,\ell)$ to the root-node $r$ in $T$ experiences a cost of at most $\frac{1}{2n_j^2}$ on the edges of this path, where $j$ denotes the label of the variable corresponding to literal $\ell$.
\end{lemma}

\begin{proof}
If $j=1$, then our construction guarantees that $l(c,\ell)$ coincides with the root-node $r$. If $j=2$, then the path from $l(c,\ell)$ to the root-node $r$ (if any) may contain the light edges of at most one literal gadget corresponding to a variable with label $1$, among which exactly one is subsidized since the assignment is balanced light. Hence, a player that uses or deviates to a strategy that contains the path from node $l(c,\ell)$ to the root-node $r$ in $T$ experiences a cost of at most $\frac{1}{n_1-3}\leq \frac{2}{n_1}=\frac{1}{2n_2^2}$.
If $j\geq 3$, then the path from $l(c,\ell)$ to the root-node $r$ (if any) may contain the light edges of at most two literal gadgets corresponding to variables with labels at most $j-1$ and $j-2$; in each of these literal gadgets, exactly one among the two light edges is subsidized since the assignment is balanced light. Hence, a player that uses or deviates to a strategy that contains the path from node $l(c,\ell)$ to the root-node $r$ in $T$ experiences a cost of at most $\frac{1}{n_{j-2}-3}+\frac{1}{n_{j-1}-3}\leq \frac{2}{n_{j-1}}=\frac{1}{2n_j^2}$.
\qed\end{proof}

From now on, we will extensively use the following definition. For each literal $\ell$, we define the set of light edges
\begin{eqnarray*}
E(\ell) &=& \left\{(u(c,\bar{\ell}),u(c,\ell)) : \mbox{ for each clause $c$ that contains $\ell$ as literal}\right\}\\
&& \cup \left\{(l(c,\bar{\ell}),u(c,\ell)) : \mbox{ for each clause $c$ that contains $\bar{\ell}$ as literal}\right\}
\end{eqnarray*}

\begin{lemma}\label{lem:ell-ell-consistency-gadget}
Consider the extension of the broadcast game on $G$ with a balanced light assignment of subsidies. Given an $\ell$-$\ell$ consistency gadget for the appearance of literal $\ell$ in clauses $c_1$ and $c_2$, the following two sentences are equivalent:
\begin{itemize}
\item The critical consistency players associated with nodes $u_1(c_1,c_2,\ell)$ and $u_2(c_1,c_2,\ell)$ have no incentive to change their strategies in $T$.
\item The light edges of the gadget that are subsidized are either those belonging to $E(\ell)$ or those belonging to $E(\bar{\ell})$.
\end{itemize}
\end{lemma}

\begin{proof}
Consider an $\ell$-$\ell$ consistency gadget that corresponds to the appearance of literal $\ell$ in clauses $c_1$ and $c_2$. Let $j$ be the label of the variable corresponding to literal $\ell$. Since the assignment is balanced light, two light edges are subsidized: one among the edges $(l(c_1,\ell),u(c_1,\bar{\ell}))$ and $(u(c_1,\bar{\ell}),u(c_1,\ell))$ of the literal gadget corresponding to the appearance of literal $\ell$ in clause $c_1$ and one among the edges $(l(c_2,\ell),u(c_2,\bar{\ell}))$ and $(u(c_2,\bar{\ell}),u(c_2,\ell))$ of the literal gadget corresponding to the appearance of literal $\ell$ in clause $c_2$.

If the subsidized edges are $(l(c_1,\ell),u(c_1,\bar{\ell}))$ and $(u(c_2,\bar{\ell}), u(c_2,\ell))$, then the critical consistency player associated with node $u_2(c_1,c_2,\ell)$ has an incentive to change her strategy in $T$. The cost she experiences in her strategy in $T$ is at least $K+\frac{1}{n_j}$. The cost she would experience by deviating to the strategy consisting of path $\langle u_2(c_1,c_2,\ell), u(c_1,\bar{\ell}), l(c_1,\ell)\rangle$ and the path from $l(c_1,\ell)$ to $r$ in $T$ would be at most $K+\frac{1}{2n_j}+\frac{1}{2n_j^2}<K+\frac{1}{n_j}$. In the left part of this inequality as well as in the inequalities below we have used Lemma \ref{lem:distance} in order to bound the cost experienced by the critical consistency player on the edges of the path from $l(c_1,\ell)$ to $r$ in $T$.

If the subsidized edges are $(u(c_1,\bar{\ell}), u(c_1,\ell))$ and $(l(c_2,\ell),u(c_2,\bar{\ell}))$, then the critical consistency player associated with node $u_1(c_1,c_2,\ell)$ has an incentive to change her strategy in $T$; due to symmetry, the argument is the same as above.

If the subsidized edges are the edges $(l(c_1,\ell),u(c_1,\bar{\ell}))$ and $(l(c_2,\ell),u(c_2,\bar{\ell}))$ that belong to set $E(\bar{\ell})$, then no critical consistency player has an incentive to change her strategy in $T$. The critical consistency player associated with node $u_1(c_1,c_2,\ell)$ experiences a cost of at most $K+\frac{1}{2n_j^2}$ while the cost she would experience by changing her strategy to the one that uses path $\langle u_1(c_1,c_2,\ell), u(c_2,\bar{\ell})), l(c_2,\ell)\rangle$ and the path in $T$ from $l(c_2,\ell)$ to $r$ would be at least $K+\frac{1}{2n_j} \geq K+\frac{1}{2n_j^2}$. Also, note that if the critical consistency player associated with node $u_1(c_1,c_2,\ell)$ deviates to a strategy that does not use the path from node $l(c_1,\ell)$ to the root-node $r$ in $T$, this would contain two heavy edges which are used only by this player for a cost of at least $2K$. Due to symmetry, the argument for the critical consistency player associated with node $u_2(c_1,c_2,\ell)$ is the same.

If the subsidized edges are the edges $(u(c_1,\bar{\ell}),u(c_1,\ell))$ and $(u(c_2,\bar{\ell}),u(c_2,\ell))$ that belong to set $E(\ell)$, then no critical consistency player has an incentive to change her strategy in $T$ either. The critical consistency player associated with node $u_1(c_1,c_2,\ell)$ experiences a cost of at most $K+\frac{1}{n_j}+\frac{1}{2n_j^2}$ while the cost she would experience by changing her strategy to the one that uses path $\langle u_1(c_1,c_2,\ell), u(c_2,\bar{\ell})), l(c_2,\ell)\rangle$ and the path in $T$ from $l(c_2,\ell)$ to $r$ would be at least $K+\frac{1}{2n_j}+\frac{1}{n_j+1} \geq K+\frac{1}{n_j}+\frac{1}{2n_j^2}$. Also, note that if the critical consistency player associated with node $u_2(c_1,c_2,\ell)$ deviates to a strategy that does not use the path from node $l(c_2,\ell)$ to the root-node $r$ in $T$, this would contain two heavy edges which are used only by this player for a cost of at least $2K$. Due to symmetry, the argument for the critical player associated with node $u_2(c_1,c_2,\ell)$ is the same.
\qed\end{proof}

\begin{lemma}\label{lem:ell-not-ell-consistency-gadget}
Consider the extension of the broadcast game on $G$ with a balanced light assignment of subsidies. Given an $\ell$-$\bar{\ell}$ consistency gadget for the appearance of literals $\ell$ and $\bar{\ell}$ in clauses $c_1$ and $c_2$, the following two sentences are equivalent:
\begin{itemize}
\item The critical consistency players associated with nodes $u_1(c_1,c_2,\ell)$ and $u_2(c_1,c_2,\ell)$ have no incentive to change their strategies in $T$.
\item The light edges of the gadget that are subsidized are either those belonging to $E(\ell)$ or those belonging to $E(\bar{\ell})$.
\end{itemize}
\end{lemma}

\begin{proof}
Consider an $\ell$-$\bar{\ell}$ consistency gadget that corresponds to the appearance of literals $\ell$ and $\bar{\ell}$ in the clauses $c_1$ and $c_2$, respectively. Let $j$ be the label of the variable corresponding to literals $\ell$ and $\bar{\ell}$. Since the assignment is balanced light, two light edges are subsidized: one among the edges $(l(c_1,\ell),u(c_1,\bar{\ell}))$ and $(u(c_1,\bar{\ell}),u(c_1,\ell))$ of the literal gadget corresponding to the appearance of literal $\ell$ in clause $c_1$ and one among the edges $(l(c_2,\bar{\ell}),u(c_2,\ell))$ and $(u(c_2,\ell),u(c_2,\bar{\ell}))$ of the literal gadget corresponding to the appearance of literal $\bar{\ell}$ in clause $c_2$.

If the subsidized edges are $(l(c_1,\ell),u(c_1,\bar{\ell}))$ and $(l(c_2,\bar{\ell}),u(c_2,\ell))$, then the critical consistency player associated with node $u_1(c_1,c_2,\ell)$ has an incentive to change her strategy in $T$. The cost she experiences in her strategy in $T$ is at least $K+\frac{1}{n_j-3}$. The cost she would experience by deviating to the strategy consisting of the path $\langle u_1(c_1,c_2,\ell), u(c_2,\ell), l(c_2,\bar{\ell})\rangle$ and the path from $l(c_2,\bar{\ell})$ to $r$ in $T$ would be at most $K+\frac{1}{n_j}+\frac{1}{2n_j^2}+\frac{1}{2n_j^2}< K+\frac{1}{n_j-3}$. In the left part of this inequality as well as in the inequalities below we have used Lemma \ref{lem:distance} in order to bound the cost experienced by the critical consistency player on the edges of the path from $l(c_1,\ell)$ to $r$ in $T$.

If the subsidized edges are $(u(c_1,\bar{\ell}),u(c_1,\ell))$ and $(u(c_2,\ell),u(c_2,\bar{\ell}))$, then the critical consistency player associated with node $u_2(c_1,c_2,\ell)$ has an incentive to change her strategy in $T$. The cost she experiences in her strategy in $T$ is at least $K+\frac{1}{n_j}$. The cost she would experience by deviating to the strategy consisting of the path $\langle u_2(c_1,c_2,\ell), u(c_1,\ell), u(c_1,\bar{\ell}), l(c_1,\ell)\rangle$ and the path from $l(c_1,\ell)$ to $r$ in $T$ would be at most $K+\frac{1}{n_j+1}+\frac{1}{2n_j^2} < K+\frac{1}{n_j}$.

If the subsidized edges are the edges $(l(c_1,\ell),u(c_1,\bar{\ell}))$ and $(u(c_2,\ell),u(c_2,\bar{\ell}))$ of the set $E(\bar{\ell})$, then no critical consistency player has an incentive to change her strategy in $T$. The critical consistency player associated with node $u_2(c_1,c_2,\ell)$ experiences a cost of at most $K+\frac{1}{n_j}+\frac{1}{2n_j^2}$. The cost she would experience by deviating to the strategy consisting of the path $\langle u_2(c_1,c_2,\ell), u(c_1,\ell), u(c_1,\bar{\ell}), l(c_1,\ell)\rangle$ and the path from $l(c_1,\ell)$ to $r$ in $T$ would be at least $K+\frac{1}{n_j-2} \geq  K+\frac{1}{n_j}+\frac{1}{2n_j^2}$. Also, note that if the critical consistency player associated with node $u_2(c_1,c_2,\ell)$ deviates to a strategy that does not use the path from node $l(c_1,\ell)$ to the root-node $r$ in $T$, this would contain two heavy edges which are used only by this player for a cost of at least $2K$. The critical consistency player associated with node $u_1(c_1,c_2,\ell)$ experiences a cost of at most $K+\frac{1}{n_j-3}+\frac{1}{2n_j^2}$. The cost she would experience by deviating to the strategy consisting of the path $\langle u_1(c_1,c_2,\ell), u(c_2,\ell), l(c_2,\bar{\ell})\rangle$ and the path from $l(c_2,\bar{\ell})$ to $r$ in $T$ would be at least $K+\frac{1}{n_j}+\frac{1}{2n_j^2}+\frac{1}{n_j+1} \geq K+\frac{1}{n_j-3}+\frac{1}{2n_j^2}$. Also, note that if the critical consistency player associated with node $u_1(c_1,c_2,\ell)$ deviates to a strategy that does not use the path from node $l(c_2,\bar{\ell})$ to the root-node $r$ in $T$, this would contain two heavy edges which are used only by this player for a cost of at least $2K$.

If the subsidized edges are the edges $(u(c_1,\bar{\ell}),u(c_1,\ell)))$ and $(l(c_2,\bar{\ell}),u(c_2,\ell))$, then no critical consistency player has an incentive to change her strategy in $T$ either. The critical consistency player associated with node $u_2(c_1,c_2,\ell)$ experiences a cost of at most $K+\frac{1}{2n_j^2}$. The cost she would experience by deviating to the strategy consisting of the path $\langle u_2(c_1,c_2,\ell), u(c_1,\ell), u(c_1,\bar{\ell}), l(c_1,\ell)\rangle$ and the path from $l(c_1,\ell)$ to $r$ in $T$ would be at least $K+\frac{1}{n_j+1} \geq  K+\frac{1}{2n_j^2}$. Also, note that if the critical consistency player associated with node $u_2(c_1,c_2,\ell)$ deviates to a strategy that does not use the path from node $l(c_1,\ell)$ to the root-node $r$ in $T$, this would contain two heavy edges which are used only by this player for a cost of at least $2K$. The critical consistency player associated with node $u_1(c_1,c_2,\ell)$ experiences a cost of at most $K+\frac{1}{n_j}+\frac{1}{2n_j^2}$. The cost she would experience by deviating to the strategy consisting of the path $\langle u_1(c_1,c_2,\ell), u(c_2,\ell), l(c_2,\bar{\ell})\rangle$ and the path from $l(c_2,\bar{\ell})$ to $r$ in $T$ would be at least $K+\frac{1}{n_j}+\frac{1}{2n_j^2}$. Also, note that if the critical consistency player associated with node $u_1(c_1,c_2,\ell)$ deviates to a strategy that does not use the path from node $l(c_2,\bar{\ell})$ to the root-node $r$ in $T$, this would contain two heavy edges which are used only by this player for a cost of at least $2K$.
\qed\end{proof}

We call a balanced light assignment such that for each variable $x$ either the edges of $E(x)$ or the edges of $E(\bar{x})$ are subsidized (i.e., if the second sentence in Lemmas \ref{lem:ell-ell-consistency-gadget} and \ref{lem:ell-not-ell-consistency-gadget} holds for every literal) a {\em consistent} balanced light assignment. Under this definition, Lemmas \ref{lem:ell-ell-consistency-gadget} and \ref{lem:ell-not-ell-consistency-gadget} yield the following corollary.

\begin{cor}\label{cor:consistency}
Consider the extension of the broadcast game on $G$ with a balanced light assignment of subsidies. The critical consistency players have no incentive to deviate from their strategy in $T$ if and only if the assignment is consistent balanced light.
\end{cor}

Note that there exists a one-to-one and onto mapping between consistent balanced light assignments of subsidies and assignments of values to the variables of $\phi$ by setting $x=1$ for every variable $x$ such that the edges in $E(x)$ are subsidized (and $x=0$ otherwise).

We proceed with Step 4.

\begin{lemma}\label{lem:clause-gadget}
Consider a light assignment of subsidies. The following are equivalent:
\begin{itemize}
\item $T$ is enforced as an equilibrium in the extension of the broadcast game.
\item (a) The assignment of subsidies is consistent balanced light. (b) For each clause $c=(\ell_1,\ell_2,\ell_3)$, at least one of the following is true:
\begin{itemize}
\item the edges of $E(\ell_1)$ are subsidized.
\item the edges of $E(\ell_2)$ are subsidized.
\item the edges of $E(\ell_3)$ are subsidized.
\end{itemize}
\end{itemize}
\end{lemma}

\begin{proof}
If $T$ is enforced as an equilibrium, then no player has any incentive to deviate from her strategy in $T$. By Lemma \ref{lem:literal-gadget} and Corollary \ref{cor:consistency}, we obtain (a). We will show that the fact that the critical clause players have no incentive to deviate from their strategy in $T$ implies (b). Assume otherwise that there exists a clause $c=(\ell_1,\ell_2,\ell_3)$ such that the edges in $E(\ell_1)\cup E(\ell_2)\cup E(\ell_3)$ are not subsidized. This means that the light edges $(u(c,\bar{\ell_1}), u(c,\ell_1))$, $(u(c,\bar{\ell_2}),u(c,\ell_2))$, and $(u(c,\bar{\ell_3}),u(c,\ell_3))$ in the three literal gadgets corresponding to the appearance of literals $\ell_1$, $\ell_2$, and $\ell_3$ in $c$ are not subsidized. Let $j_1$, $j_2$, and $j_3$ be the labels of the variables corresponding to literals $\ell_1$, $\ell_2$, and $\ell_3$. Then, the cost the critical clause player associated with node $v(c)$ experiences on her strategy in $T$ is at least
$K+\frac{1}{n_{j_1}-3}+\frac{1}{n_{j_2}-3}+\frac{1}{n_{j_3}-3}$. Hence, she has an incentive to deviate and use the direct edge $(v(c),r)$ of weight $K+\frac{1}{n_{j_1}}+\frac{1}{n_{j_2}-3}+\frac{1}{n_{j_3}-3}$ which contradicts our assumption.

On the other hand, if (a) and (b) are true, we will show that no player has any incentive to deviate from her strategy in $T$. Again, by Lemma \ref{lem:literal-gadget} and Corollary \ref{cor:consistency}, we have that (a) implies that critical literal players and critical consistency players have no incentive to deviate while Lemma \ref{lem:non-critical} implies that non-critical players have no incentive to deviate. We will show that (b) implies that no critical clause player has an incentive to deviate either. Indeed, consider a clause $c=(\ell_1,\ell_2,\ell_3)$ and let $j_1$, $j_2$, and $j_3$ be the labels of the variables corresponding to literals $\ell_1$, $\ell_2$, and $\ell_3$. (b) implies that one among the three light edges $(u(c,\bar{\ell_i}),u(c,\ell_i))$ for $i\in \{1,2,3\}$ in the three literal gadgets corresponding to the appearance of literal $\ell_i$ is subsidized and, due to (a), if a light edge $(u(c,\bar{\ell_i}), u(c,\ell_i))$ is not subsidized, then the light edge $(l(c,\ell_i), u(c,\bar{\ell_i})$ is subsidized. Hence, the cost of the critical clause player associated with node $v(c)$ is at most $K+\max\left\{\frac{1}{n_{j_1}}+\frac{1}{n_{j_2}-3}+\frac{1}{n_{j_3}-3}, \frac{1}{n_{j_1}-3}+\frac{1}{n_{j_2}}+\frac{1}{n_{j_3}-3}, \frac{1}{n_{j_1}-3}+\frac{1}{n_{j_2}-3}+\frac{1}{n_{j_3}}\right\} \leq K+\frac{1}{n_{j_1}}+\frac{1}{n_{j_2}-3}+\frac{1}{n_{j_3}-3}$ and, hence, this player has no incentive to use edge $(v(c),r)$. Also, she has no incentive to deviate to another strategy which includes edge $(v(c),u(c,\ell_3))$ but not the path from $u(c,\ell_3)$ to $r$ in $T$. Any such path contains two heavy edges which would be used only by the critical clause player associated with node $v(c)$ and, hence, her cost would be at least $2K$.
\qed\end{proof}

Note that the mapping mentioned above is a one-to-one and onto mapping between consistent balanced light assignments of subsidies and truthful assignments of $\phi$. We obtain the following corollary.
\begin{cor}\label{cor:T-sat-eq}
There exists a light assignment of subsidies that enforces $T$ as an equilibrium in the extension of the broadcast game on $G$ if and only if $\phi$ is satisfiable.\end{cor}

Hence, if $\phi$ is not satisfiable, then the minimum amount of all-or-nothing subsidies necessary to enforce $T$ as an equilibrium in the extension of the broadcast game is at least $K$ (some heavy edge has to be subsidized). We conclude that distinguishing between whether all-or-nothing subsidies of cost (at most) $3|C|$ are sufficient or subsidies of cost at least $K$ are necessary in order to enforce $T$ as an equilibrium in the extension of the broadcast game is NP-hard. Theorem \ref{thm:hardness-all-or-nothing} follows by setting $K$ to be arbitrarily large compared to $|C|$.
\qed
\end{proof}

Theorem \ref{thm:hardness-all-or-nothing} probably indicates that the only approximation guarantee we should hope for all-or-nothing SNE is to bound the amount of subsidies as a constant fraction of the weight of an optimal design. The next statement implies that significantly more subsidies may be necessary compared to the standard version of SNE in order to enforce a minimum spanning tree as an equilibrium.

\begin{theorem}\label{thm:all-or-nothing-lb}
For every $\epsilon>0$, there exist a broadcast game on a graph $G$ and a minimum spanning tree $T$ of $G$ such that the cost of any all-or-nothing subsidy assignment that enforces $T$ as an equilibrium in the extension of the broadcast game is at least $\left(\frac{e}{2e-1}-\epsilon\right)\wgt(T)$.
\end{theorem}

\begin{proof}
We will define a graph $G$ with $n+1$ nodes which has a minimum spanning tree that consists of a path $\langle r, v_1, v_2, ..., v_n\rangle$. Let $x=\left(n-n/e+1\right)^{-1}$. Edges $(r,v_1)$ and $(v_i, v_{i+1})$ for $i = 1, ..., n-2$ have weight $x$. Edge $(v_{n-1},v_n)$ has weight $1$. The graph contains two additional edges: edge $(r,v_{n-1})$ has weight $x$ and edge $(r,v_n)$ has weight $1$. If we do not put subsidies on the edge $(v_{n-1},v_n)$, then we have to put subsidies on each of the remaining edges in the path in order to guarantee that the player associated with node $v_{n}$ has no incentive to use the direct edge $(v_{n},r)$, i.e., a total amount of $(n-1)x$ as subsidies. If we put subsidies on the edge $(v_{n-1},v_n)$, we still have to guarantee that the player associated with node $v_{n-1}$ has no incentive to deviate to the direct edge $(v_{n-1},r)$. Using the same reasoning as in the proof of Theorem \ref{thm:lb}, we will need an amount of at least $(n/e-2)x$ as subsidies on the edges of the path $\langle r, v_1, v_2, ..., v_{n-1}\rangle$, for a total of $1+(n/e-2)x$. Due to the definition of $x$, we have that the amount of subsidies is at least $\frac{n-1}{n-n/e+1}$ in both cases while the total weight of $T$ is $\frac{2n-n/e}{n-n/e+1}$. The bound follows by selecting $n$ to be sufficiently large.
\qed\end{proof}

\section{Open problems}\label{sec:open}
Our work has revealed several open questions. Concerning the particular results obtained, it is interesting to design a combinatorial algorithm for SNE which, on input a graph $G$ and a minimum spanning tree $T$ on $G$, enforces $T$ as an equilibrium on the corresponding broadcast game using minimum subsidies. Lemma \ref{lem:lp-simple} may be helpful in this direction. For the integral version of SNE, we have left open the question whether it is always possible to enforce a given minimum spanning tree as an equilibrium in a broadcast game using all-or-nothing subsidies of cost strictly smaller than the weight of a minimum spanning tree. Given our negative result in Theorem \ref{thm:hardness-all-or-nothing}, this is probably the only approximation that makes sense. It is tempting to conjecture that our lower bound is tight, i.e., there is an algorithm that always uses a fraction of at most $\frac{e}{2e-1}$ of the weight of the minimum spanning tree as subsidies in order to do so.

Approximating SND would also be interesting. Given the known hardness statements (e.g., \cite{S10}) or the lack of positive results concerning the complexity of computing equilibria, this is a far more challenging goal. A concrete question for SND instances consisting of broadcast games could be the following: can we compute in polynomial time an equilibrium tree using subsidies of cost at most an $\alpha$ fraction of the weight of the minimum spanning tree? Our results (Theorems \ref{thm:SNE-P} and \ref{thm:upper-bound}) indicate that the answer is clearly positive if $\alpha\geq 1/e$. Is this also possible if $\alpha$ is an arbitrarily small constant? Definitely, more general instances of SND (e.g., involving multicast games) are challenging as well. Finally, variations of SNE and SND that consider deviations of coalitions of players (as opposed to unilateral deviations), players with different demands \cite{Alb08,CR06}, or different cost sharing protocols \cite{CRV08} deserve investigation.

\small\bigskip\noindent{\bf Acknowledgments.} We thank Edith Elkind, Ning Chen, Nick Gravin, and Alex Skopalik for helpful discussions.

\end{document}